\documentclass[12pt,draftcls,onecolumn]{IEEEtran}

\ifCLASSINFOpdf
\else
\fi
\usepackage{amsmath}
\usepackage{amssymb}
\usepackage{amsthm}
\usepackage{mathtools}
\usepackage{graphics}
\usepackage{epsfig}
\usepackage{graphicx}
\usepackage{epstopdf}
\usepackage{subfigure}
\usepackage{algorithm}
\usepackage{algpseudocode}
\usepackage{color}

\def \e{\mathbf{e}}
\def \u{\mathbf{u}}

\def \0{\mathbf{0}}
\def \1{\mathbf{1}}

\def \B{\mathbf{B}}
\def \C{\mathbf{C}}

\def \d{\mathbf{d}}

\def \H{\mathbf{H}}
\def \I{\mathbf{I}}
\def \K{\mathbf{K}}
\def \M{\mathbf{M}}
\def \P{\mathbf{P}}

\def \R{\mathbf{R}}
\def \S{\mathbf{S}}

\def \bK{\bar{\mathbf{K}}}
\def \hK{\hat{\mathbf{K}}}
\def \hG{\hat{G}}
\def \bG{\bar{G}}

\def \hg{\hat{g}}
\def \bg{\bar{g}}

\def \cP{\mathcal{P}}

\def \cF{\mathcal{F}}

\def \cK{\mathcal{K}}

\def \cA{\mathcal{A}}

\def \cU{\mathcal{U}}

\def \bR{\mathbb{R}}

\def \E{\textrm{E}}
\def \Cov{\textrm{Cov}}
\def \Var{\textrm{Var}}
\def \LQ{\textrm{LQ}}
\def \MV{\textrm{MV}}

\newtheorem{lemma}{Lemma}
\newtheorem{proposition}{Proposition}
\newtheorem{theorem}{Theorem}

\theoremstyle{definition}
\newtheorem{assumption}{Assumption}
\newtheorem{example}{Example}

\begin{document}
%
\title{Explicit Solution for Constrained Stochastic
Linear-Quadratic Control with Multiplicative Noise}
%
%
%

\author{Weiping Wu ,~\IEEEmembership{Student Member,~IEEE,}
        Jianjun Gao,~\IEEEmembership{Member,~IEEE,}\\
        Duan Li,~\IEEEmembership{Senior~Member,~IEEE,}~
        Yun Shi

\thanks{W.~P.~Wu is with the Automation Department, Shanghai Jiao Tong University, Shanghai, China (email: godream@sjtu.edu.cn).}
\thanks{J.~J.~Gao is with School of Information Management and Engineering, University of Finance and Economics, Shanghai, China (e-mail:gao.jianjun@shufe.edu.cn).}
\thanks{D.~Li is with the Department of Systems Engineering  Engineering Management, The Chinese University of Hong Kong, Hong Kong (email: dli@se.cuhk.edu.hk).}
\thanks{Y.~Shi is with the School of Management, Shanghai University, Shanghai China, (email:y\_shi@shu.edu.cn).}
\thanks{This research work was partially supported by National Natural Science Foundation of China under grant 61573244, and by the Research Grants Council of Hong Kong under grant 14213716. }
}
%
%

\markboth{Journal of \LaTeX\ Class Files,~Vol.~xx, No.~x, August~2015}%
{Shell \MakeLowercase{\textit{et al.}}: Bare Demo of IEEEtran.cls for IEEE Journals}
%



\maketitle

\begin{abstract}
We study in this paper a class of constrained linear-quadratic (LQ) optimal control problem formulations for the scalar-state stochastic system with multiplicative noise, which has various applications, especially in the financial risk management. The linear constraint on both the control and state variables considered in our model destroys the elegant structure of the conventional LQ formulation and has blocked the derivation of an explicit control policy so far in the literature. We successfully derive in this paper the analytical control policy for such a class of problems by utilizing the state separation property induced from its structure. We reveal that the optimal control policy is a piece-wise affine function of the state and can be computed off-line efficiently by solving two coupled Riccati equations. Under some mild conditions, we also obtain the stationary control policy for infinite time horizon. We demonstrate the implementation of our method via some illustrative examples and show how to calibrate our model to solve dynamic constrained portfolio optimization problems.
\end{abstract}


%
\IEEEpeerreviewmaketitle

\section{Introduction}
\IEEEPARstart{W}e study in this paper the constrained linear-quadratic (LQ) control problem for the discrete-time stochastic scalar-state system with multiplicative noise. The past few years have witnessed intensified attention on this subject, due to its promising applications in different areas, including dynamic portfolio management, financial derivative pricing, population model, and nuclear heat transfer (see, e.g., \cite{GaoLiCuiWang:2015}\cite{Costabook:2007}\cite{Primb:2009}).

There exist in the literature various studies on the estimation and control problems of systems with a multiplicative noise \cite{Basin:2006}\cite{Gershon:2006}. As for the LQ type of stochastic optimal control problems with multiplicative noise, investigations have been focused on the LQ formulation with indefinite penalty matrices on control and state variables for both continuous-time and discrete-time models (see, e.g., \cite{LimZhou:1999}\cite{LiZhouRami:2003}\cite{Zhu:2005}\cite{RamiChenMooreZhou:2001}\cite{Costa:2007}).
One interesting finding is that even when the penalty matrices for both state and control are indefinite, this kind of models with multiplicative noise is still well-posed under some conditions. One important application of this kind of models arises in the dynamic mean-variance (MV) portfolio analysis \cite{LiNg:2000} \cite{ZhouLi:2000}, which generalizes the Markowitz's classical work \cite{Markowitz:1969} on static portfolio selection. Please see, e.g., \cite{GaoLiCuiWang:2015}\cite{CuiLiLi:2014} for some detailed surveys on this subject which
has grown significantly in recent years.


One prominent attractiveness of the LQ type of optimal control models is its explicit control policy which can be derived by solving the correspondent Riccati equation. However, in real applications, due to some physical limits, consideration of the risk or the economic regulation restrictions, some constraints on the control variables have to be taken into the consideration. Unfortunately, when some control constraints are involved, except for a few special cases, there is hardly a closed-form control policy for the constrained LQ optimal control model. As for the deterministic LQ control problem, Gao et al. \cite{GaoLi:2011} investigate the LQ model with cardinality constraints on the control and derive a semi-analytical solution. For the model with inequality constraints on state and control, Bemporad et al. \cite{Bemporad:2002} propose a method by using a parametric programming approach to compute an explicit control policy. However, this method may suffer a heavy computational burden when the size of the problem is increasing. When the problem only involves the positivity constraint for the control, some scholarly works \cite{Campbell:1982}\cite{HeemelsEijndhovenStoorvogel:1998} provide the optimality conditions and some numerical methods to characterize the optimal control policy. Due to the difficulties in characterizing the explicit optimal control, it is more tractable to develop some approximated control policy by using, for example, the Model Predictive Control (MPC) approach \cite{Mayne:2014}\cite{Kouvaritakis:2014}. The main idea behind the MPC is to solve a sub-problem with finite horizon at each time period for an open-loop control policy and implement such a control in a fashion of rolling horizon.  As only a static optimization problem needs to be solved in each step, this kind of model can deal with general convex constraints. As for the stochastic MPC problem with multiplicative noise, Primbs et al.\cite{Primb:2009} propose a method by using the semi-definite programming and supply a condition for the system stability. Bernardini and Bemporad \cite{BernardiniBemporad:2012} study a similar problem with discrete random scenarios and propose some efficient computational methods by solving quadratically constrained quadratic programming problems off-line. Patrinos et al. \cite{Patrinos:2014} further extend such a method to solve the problem with Markovain jump. Readers may refer \cite{ChanHsuSethi:2002}\cite{Mesbah:2016} for more complete surveys of MPC for stochastic systems.

The current literature lacks progress in obtaining an explicit solution for the constrained stochastic LQ type optimal control problems. However, some promising results have emerged recently for dynamic MV portfolio selection, a special class of such problems. Li et al. \cite{LiZhouLim:2002} characterize the analytical solution of the continuous-time MV portfolio selection problem with no shorting by using the viscosity solution of the partial differential equation.  The work by Hu and Zhou \cite{HuZhou:2005} solve the cone constrained continuous-time LQ control problem with a scalar state by using  the backward stochastic differential equation (BSDE) approach. Cui et al. \cite{CuiGaoLiLi:2014} \cite{CuiLiLi:2015} solve the discrete-time version of this type of problems with no shorting constraint and cone constraints, respectively.  Note that the models studied in \cite{CuiGaoLiLi:2014} \cite{CuiLiLi:2015} are just some special cases of our model studied in this paper. Gao et al. \cite{GaoLiCuiWang:2015} derive the solution for the dynamic portfolio optimization model with cardinality constraint with respect to the active periods in time.

In this paper, we focus on the constrained LQ optimal control for the scalar-state stochastic system with multiplicative noise. The contributions of our work include several aspects. First, we derive the analytical control law of this type of problem with a general class of general linear constraints, which goes beyond the cone constraints studied in \cite{HuZhou:2005}\cite{CuiLiLi:2015}. This general constraint also includes positivity and negativity constraints, state-dependent upper and lower bound constraints as its special cases. We show that the control policy is a piece-wise affine function with respect to the state variable, which can be characterized by solving  two coupled Riccati equations with two unknowns. Second, we extend such results to the problem with infinite horizon. We provide the condition on the existence of the solution for the correspondent algebraic Riccati equations and show that the closed-loop system is asymptotically stable under the stationary optimal control. Besides the theoretical study, we illustrate how to use this kind of models to solve the constrained dynamic mean-variance portfolio optimization.

The paper is organized as follows. Section \ref{se_fomulation} provides the formulations of the stochastic LQ control problem with control constraints for both finite and infinite horizons. Section \ref{se_sol_P} and Section \ref{se_sol_Pinf} develop the explicit solutions for these two problems, respectively. Section \ref{se_mv} illustrates how to apply our method to solve the constrained dynamic mean-variance portfolio selection problem. Section \ref{se_exampl} presents some numerical examples to demonstrate the effectiveness of the proposed solution schemes. Finally, Section \ref{se_conclusion} concludes the paper with some possible further extensions.

\textit{Notations} The notations $\0_{n\times m }$ and $\I_{n}$ denote the $n \times m$ zero matrix and the $n\times n$ identity matrix, respectively, $\R\succeq 0$ ($\R\succ 0$) denotes a positive semidefinite (positive definite) matrix, and $\bR$ ($\bR_+$) denotes the set of real (nonnegative real) numbers. We denote by $\1_{\mathcal{A}}$ the indicator function such that $\1_{\mathcal{A}}=1$ if the condition $\mathcal{A}$ holds true and $\1_{\mathcal{A}}=0$ otherwise. Let $\E_t[\cdot]$ be the conditional expectation $\E[\cdot|\cF_t]$ with $\cF_t$ being the filtration (information set) at time $t$. For any problem $(\cP)$, we use $v(\cP)$ to denote its optimal objective value.

\section{Problem Formulations}\label{se_fomulation}

\subsection{Problem with Finite Time Horizon}\label{sse_finite}
In this work, we consider the following scalar-state discrete-time linear stochastic dynamic system,
\begin{align}
x_{t+1} = A_t x_t + \B_t \u_t,~t = 0, 1, \cdots, T-1, \label{def_xt}
\end{align}
where $T$ is a finite positive integer number, $x_t$ $\in $ $\bR$ is the state with $x_0$ being given, $\u_t$ $\in$ $\bR^n$ is the control vector, $\{A_t\in\bR\}|_{t=0}^{T-1}$ and $\{\B_t \in\bR^{1\times  n}\}|_{t=0}^{T-1}$ are random system parameters. In the above system model, all the randomness are modeled by a completely filtrated probability space $\{\Omega,\{ \cF_t\}|_{t=0}^T, \mathbb{P}\}$, where $\Omega$ is the event set, $\mathcal{F}_t$ is the filtration of the information available at time $t$ with $\cF_0=\{\emptyset, \Omega\}$, and $\mathbb{P}$ is the probability measure. More specifically, at any time $t\in \{1, \cdots, T\}$, the filtration $\cF_t$ is the smallest sigma-algebra generated by the realizations of $\{A_k\}|_{k=0}^{t-1}$ and $\{\B_k\}|_{k=0}^{t-1}$. That is to say, in our model, the random parameters $A_t$ and $\B_t$ are $\cF_{t+1}$ measurable for any $t=0,\cdots, T-1$.\footnote{The random structure we adopt in this paper has been commonly used in the area of financial engineering \cite{FollmerSchied:2004}. This kind of models is very general, since it does not need to specify the particular stochastic processes which $A_t$  and $\B_t$ follow.} To simplify the notation, we use $\E_t[\cdot]$ to denote the conditional expectation with respect to filtration $\cF_t$. To guarantee the well-posedness of the model, we assume all $A_t$ and $\B_t$ are square integrable, i.e., $\E_t[ |A_t|^2]$$<$$\infty$ and $\E_t[\|\B_t\|^2]$$<$ $\infty$ for all $t$ $=$ $0$, $\cdots$, $T-1$.

Note that the above stochastic dynamic system model is very general. For example, it covers the traditional stochastic uncertain systems with a scalar state space and multiplicative noise \cite{Primb:2009}\footnote{In \cite{Primb:2009} and many related literatures, the stochastic dynamic system is modeled as $x_{t+1}$ $=$ $ Ax_k+Bu_k$ $+$ $\sum_{j=1}^q[C_jx_k+D_ju_k]\omega_{k}^j$, where $\omega_{k}^j$ are i.i.d random variables for different $k$ with zero mean and $\E[(\omega_k^j)^2]=1$ and $\E[\omega_k^i\omega_k^j]=0$, if $i\not=j$.}. It also covers the cases in which $\{A_t,\B_t\}|_{t=0}^{T-1}$ are serially correlated stochastic processes such as Markov Chain models \cite{Costa:2007}\cite{Costa:2012} or the conventional time series models, which have important applications in financial decision making \cite{GaoLiCuiWang:2015}\cite{Costa:2008}. As for the control variables $\{\u_t\}|_{t=0}^{T-1}$, they are required to be $\cF_t$-measurable, i.e., the control $\u_t$ at time $t$ only depends on the information available up to time $t$. Furthermore, motivated by some real applications, we consider the following general control constraint set,
\begin{align}
\cU_t(x_t)=\{&\u_t|\textrm{$\u_t$ is $\cF_t$-measurable},~ \H_t \u_t \leq \d_t |x_t|\}, \label{def_Ut}
\end{align}
for $t=0,\cdots,T-1$, where $\H_t\in \bR^{m\times n}$ and $\d_t\in \bR^m$ are deterministic matrices and vectors\footnote{Since both $x_t$ and $u_t$ are random variables (except $x_0$) for all $t>0$, the inequalities given in (\ref{def_Ut}) should be held \textit{almost surely}, i.e., the inequalities are held for the cases with non-zero probability measure. To simplify the notation, we do not write out the term `almost surely' explicitly in this paper.}. Note that set ($\ref{def_Ut}$) enables us to model various control constraints as evidenced from the following:
\begin{itemize}
  \item the nonnegativity (or nonpositivity)  constraint case, $\u_t\geq \0_{n\times 1}$ (or $\u_t\leq \0_{n\times 1}$) by setting $\H_t=-\I_{n}$ (or $\H_t=\I_{n}$) and $\d_t=\0_{n\times 1}$  in (\ref{def_Ut});

  \item the constrained case with state-dependent upper and lower bounds, $\underline{\d}_t|x_t|$ $\leq$ $\u_t$ $\leq$ $\overline{\d}_t |x_t|$ for some $\underline{\d}_t$$\in$$ \R^{n}$ and $\overline{\d}_t$$ \in $$\R^n$ by setting $\H_t$ $=$ $\left(\begin{array}{c}
                                                               \I_{n} \\
                                                               -\I_{n }
                                                             \end{array}\right)
      $ and $\d_t$ $=$ $\left(
                    \begin{array}{c}
                      \overline{\d}_t \\
                      -\underline{\d}_t \\
                    \end{array}
                  \right)$;
  \item the general cone constraint case, $\H_t \u_t \geq \0_{m\times 1}$, for some $\H_t$;
  \item unconstrained case, $\u_t\in \bR^n$, by setting $\H_t=\0_{m\times n}$ and $\d_t=\0_{m\times 1}$.
\end{itemize}
To model the cost function, we introduce the following deterministic parameters,  $\{\R_t \in \bR^{n\times n}|\R_t\succeq 0\}|_{t=0}^{T-1}$, $\{\S_t\in \bR^{n}\}|_{t=0}^{T-1}$ and $\{ q_t\geq 0\}|_{t=0}^T$, which can be further written in more compact forms, $\C_t$ $:=$ $\left(\begin{array}{cc}
\R_t          & \S_t\\
\S_t^{\prime} & q_t
\end{array}\right)$ for $t$ $=$ $0$, $\cdots$, $T-1$ and $\C_T$ $:=$ $\left(\begin{array}{cc}
\0_{n\times n} & \0_{n\times 1}\\
\0_{1\times n} & q_T
\end{array}\right)$.
Overall, we are interested in the following class of inequality constrained stochastic LQ control problem (ICLQ),
\begin{align}
(\cP_{\LQ}^T)&~\min_{\{\u_t\}|_{t=0}^{T-1}}~\E_0 \left[ \sum_{t=0}^{T}
\left(\begin{array}{cc}
\u_t\\
x_t
\end{array}\right)^{\prime}\C_t
\left(\begin{array}{cc}
\u_t\\
x_t
\end{array}\right)~\right] \label{def_J}\\
{\rm s.t.}~&\{x_t,\u_t\}~\textrm{satisfies (\ref{def_xt}) and (\ref{def_Ut}) for $t=0,\cdots,T-1$}.\notag
\end{align}
To solve problem $(\cP_{\LQ}^T)$, we need the following assumption.

\begin{assumption}\label{asmp_psd}
$\C_t\succeq 0$ for $t=0,\cdots, T$, and $\Cov_t[\B_t]\succ 0$ for all $t=0,\cdots, T-1$.
\end{assumption}
Assumption \ref{asmp_psd} guarantees the convexity of problem $(\cP_{\LQ}^T)$. Assumption \ref{asmp_psd} can be regarded as a generalization of the one widely used in the mean-variance portfolio selection, which requires $\Cov_t[\B_t]\succ 0$ (Please see, for example, \cite{GaoLiCuiWang:2015} for detailed discussion). Also Assumption \ref{asmp_psd} is looser than the one used in \cite{ChenKouWang:2016}, which requires path-wise positiveness of the random matrix. Note that, since $\Cov_t[\B_t]$ $:=$ $\E_t[\B_t^{\prime}\B_t]$ $-$ $\E_t[\B_t^{\prime}]\E_t[\B_t]$,  Assumption \ref{asmp_psd} implies $\E_t[\B_t^{\prime}\B_t]$$\succ$$ 0$ for $t=$ $0$, $\cdots$, $T-1$.

\subsection{Problem with Infinite Time Horizon}\label{sse_infinite}
We are also interested in a variant of problem $(\cP_{\LQ}^T)$ with infinite time horizon. More specifically, we want to investigate the stationary control policy and long-time performance for infinite time horizon. In such an infinite time horizon, we assume that all the random parameters, $\{A_k\}|_{k=0}^{\infty}$ and $\{\B_k\}|_{k=0}^{\infty}$, are independent and identically distributed (i.i.d) over different time periods. Thus, we drop the index of time and simply use the random variable $A$ and random vector $\B$ to denote the random parameters, which leads to the following simplified version of dynamic system (\ref{def_xt}),
\begin{align}
x_{t+1}=A x_t +\B\u_t,~~~t=0,1,\cdots,T-1,\label{def_xt_inf}
\end{align}
where the system parameters $A$ and $\B$ are random with known joint distribution. As for the constraint (\ref{def_Ut}), we also assume that all $\H_t$ and $\d_t$ are fixed at $\H$ and $\d$, respectively, which leads to the following simplified version of constraint (\ref{def_Ut}),
\begin{align}
\cU_t(x_t)=\{\u_t~|~\textrm{$\u_t \in \bR^n$ },~ \H \u_t \leq \d |x_t|\}, \label{def_Ut_inf}
\end{align}
for $t=$ $0$, $1$, $\cdots$, $\infty$. To guarantee the feasibility of the constraint, we impose the following assumption.
\begin{assumption}\label{asmp_Ut}
The set $\cU_t(0)$ $=$ $\{ \u\in \bR^n |~ \H \u \leq \0_{m\times 1}\}$ is nonempty.
\end{assumption}
Note that $\cU_t(0)$ is independent of $x_t$ and Assumption \ref{asmp_Ut} implies that the feasible set $\cU_t(x_t)$ is nonempty for any $|x_t|>0$. We also set all penalty matrices $\C_t$, $t$ $=$ $0$, $\cdots$, $\infty$, at $\C$ $:=$ $\left(\begin{array}{cc}
\R          & \S\\
\S^{\prime} & q
\end{array}\right)$. We consider now the following ICLQ problem with infinite time horizon,
\begin{align*}
(\cP_{\LQ}^{\infty})~&~\min_{\{\u_t\}|_{t=0}^{\infty}}~\E \left[ \sum_{t=0}^{\infty}
\left(\begin{array}{cc}
\u_t\\
x_t
\end{array}\right)^{\prime}\C
\left(\begin{array}{cc}
\u_t\\
x_t
\end{array}\right) \right]\\
\textrm{s.t.}~&~\{x_t,\u_t\} \textrm{~satisfies (\ref{def_xt_inf}) and (\ref{def_Ut_inf}) for $t=0,\cdots,\infty$}.
\end{align*}
Note that the expectation in problem $(\cP_{\LQ}^{\infty})$ is an unconditional expectation, since $\{A, \B\}$  are independent over time.\footnote{For model $(\cP_{\LQ}^{\infty})$, since $\{A, \B\}$  are independent over time, we just simplify the notation $\E_0[\cdot]$ to $\E[\cdot]$.} For problem $(\cP_{\LQ}^{\infty})$, we need to strengthen Assumption  \ref{asmp_psd} by requiring $\C$ to be positive definite as follows,
\begin{assumption}\label{asmp_C}
$\C\succ 0$ and $\Cov[\B]$ $=$ $\E[\B^{\prime}\B]$
$-$ $\E[\B^{\prime}]\E[\B]$ $\succ$ $0$
\end{assumption}

\section{Solution scheme for problem ($\cP_{\LQ}^T$)}\label{se_sol_P}

In this section, we first reveal an important result of state separation for our models and then develop the solution for problem $(\cP_{\LQ}^T)$.

\subsection{State Separation Theorem}\label{sse_state_separation}
To derive the explicit solution of problem $(\cP_{\LQ}^T)$, we first introduce the following sets associated with the control constraint set $\cU_t(x_t)$, $\mathcal{K}_t$$:=$$\{$$\K$$\in$$\bR^n$$~|~$$\H_t \K $$\leq$$ \d_t$$\}$,
for $t=0,\cdots, T-1$. For problem $(\cP_{\LQ}^T)$, we further introduce three auxiliary optimization problems, $(\cP_t)$, $(\hat{\cP}_t)$ and $(\bar{\cP}_t)$ for  $t={0,1,\cdots,T-1}$  as follows,
\begin{align}
(\cP_t) ~&~ \min_{\u \in \cU_t}~g_t(\u,x,y,z), \label{def_Pt}\\
(\hat{\cP}_t)~&~\min_{\K \in \cK_t}~\hg_t(\K, y, z), \notag\\
(\bar{\cP}_t)~&~\min_{\K \in \cK_t}~\bg_t(\K, y,z), \notag
\end{align}
where $x\in  \bR$ is $\cF_t$-measurable random variable, $y\in \bR_+$ and $z \in \bR_+$ are $\cF_{t+1}$-measurable random variables and $g_t(\u,x,y,z):$ $\bR^n\times\bR \times \bR_+ \times \bR_+$ $\rightarrow$ $\bR$, $\hg_t(\K,y,z):$ $\bR^n \times \bR_+ \times \bR_+$ $\rightarrow$ $\bR$ and $\bg_t(\K,y,z):$ $\bR^n \times \bR_+ \times \bR_+$ $\rightarrow$ $\bR$, are respectively defined as
\begin{align}
&g_t(\u,x, y,z):=\E_t \Big[ \left(\begin{array}{cc}
\u\\
x
\end{array}\right)^{\prime}\C_t
\left(\begin{array}{cc}
\u\\
x
\end{array}\right)+(A_tx+\B_t\u)^2 \notag\\
&~~~~\times \big(y\1_{ \{A_tx +\B_t\u \geq 0\}}
+z\1_{\{A_tx+\B_t \u <0\}}\big)\Big], \label{def_gt}\\
&\hg_t(\K, y, z):=\E_t \Big[ \left(\begin{array}{cc}
		\K\\
		1
	\end{array}\right)^{\prime}\C_t
	\left(\begin{array}{cc}
		\K\\
		1
	\end{array}\right)+(A_t+\B_t \K)^2\notag\\
&~~~~\times \big( y\1_{\{A_t+\B_t \K \geq 0\}}
+ z\1_{\{A_t+\B_t \K <0\}}\big)\Big], \label{def_hgt}\\
&\bg_t(\K,y,z):=
	\E_t \Big[ \left(\begin{array}{cc}
	-\K\\
	1
	\end{array}\right)^{\prime}\C_t
	\left(\begin{array}{cc}
	-\K\\
	1
	\end{array}\right)+(A_t-\B_t  \K)^2\notag\\
&~~~~\times\big(y\1_{\{A_t-\B_t \K \leq 0\}}
+z\1_{\{A_t-\B_t \K > 0\}} \big)\Big].\label{def_bgt}
\end{align}
Since $\C_t \succeq 0$, it always holds true that $g_t(\u,x,y,z)\geq 0$, $\hg(\K,y,z)\geq 0$ and $\bg(\K,y,z)\geq 0$. Before we present the main result, we present the following lemma.
\begin{lemma}\label{lem_convex}
The function $g_t(\u,x,y,z)$ is convex with respect to $\u$, and both $\hg_t(\K, y,z )$ and $\hg_t(\K,y,z)$ are convex functions with respect to $\K$.
\end{lemma}
\begin{proof}
Checking the gradient and Hessian matrix of function ${g}_t(\u, x, y, z)$ with respect to $\u$ gives rise to\footnote{In the following part, we need to compute the partial derivative of
function $\E_t[f(\u)]$ with respect to $\u$. Under some mild conditions, we can always compute the derivative by taking the derivative of $f(\u)$ inside the expectation first. The technical condition guarantees this exchangeability of differentiation and expectation can be found in Theorem 1.21 of \cite{Kellenberg}.}
\begin{align}
&\nabla_{\u}{g}_t(\u,x,y,z)=2\E_t \Big[(\R_t\u+\S_t x)
	+ \B_t^{\prime}\big(A_t x+\B_t \u \big) \notag \\
	&~\times~\big(y \1_{\{A_t x+\B_t \u  \geq 0\}}
    + z\1_{\{A_tx +\B_t \u<0\}}\big) \Big] \label{def_nabla_gt}
\end{align}
and
\begin{align}
&\nabla_{\u}^2 {g}_t(\u,x,y,z)=2\E_t \Big[ \big(\R_t+y \B_t^{\prime}\B_t
\big)\1_{\{A_tx+\B_t \u \geq 0\}} \notag \\
&+\big(\R_t + z\B_t^{\prime}\B_t\big)\1_{\{A_tx +\B_t\u <0\}}\Big]\notag\\
&=\E_t\Big[\R_t + \big( y\1_{\{y\leq z\}}
+ y\1_{\{y>z\}}\big )\B_t^{\prime}\B_t\1_{\{A_tx+\B_t\u \geq 0\}}\notag\\
&+\big(z\1_{\{y\leq z\}} + z\1_{\{y>z\}}\big)\B_t^{\prime}\B_t\1_{\{A_tx +\B_t \u< 0\}} \Big]
\label{lem_convex_eq1}.
\end{align}
Note that the following two inequalities always hold,
\begin{align}
&y\1_{\{y >z\}} >z\1_{\{y>z\}}, z\1_{\{y\leq z \}} \geq y\1_{\{y\leq z\}}.\label{lem_convex_eq2}
\end{align}
Using the above inequalities and noticing that $y$ $>$ $0$ and $z$ $>$ $0$, we can rearrange the terms of (\ref{lem_convex_eq1}) to reach the conclusion of $\nabla_{\u}^2 {g}_t(\u,x,y,z)$ $\succeq$ $\E_t\big[\M_t\big]$, where
\begin{align}
\M_t=\R_t + \big(y\1_{\{y\leq z\}} +z\1_{\{y >z\}}\big)\B_t^{\prime}\B_t. \label{def_Mt}
\end{align}
Assumption \ref{asmp_psd} implies $\E_t[\B_t^{\prime}\B_t]$ $\succ$ $ 0$  and $\R_t$ $\succeq$ $0$. Moreover, the term $\big(y \1_{\{y \leq z\}}$ $+$ $z\1_{\{y>z\}}\big)$ is a positive random variable. These conditions guarantee $\E_t[\M_t]\succ 0$, which further implies $\nabla_{\u}^2 {g}_t(\u, x,y,z)\succ 0$. That is to say, $g_t(\u,x,y,z)$ is a strictly convex function of $\u$. As we can also apply the same procedure to $\hg_t(\K, y, z)$ and $\bg_t(\K, y, z)$ to prove their convexity with respect to $\K$, we omit the detailed proofs here.
\end{proof}
One immediate implication of  Lemma \ref{lem_convex} is that all $(\cP_t)$, $(\hat{\cP}_t)$ and $(\bar{\cP}_t)$ are convex optimization problems, as their objective functions are convex and their constraints are linear with respect to the decision variables. We can see that problem $(\cP_t)$ depends on random state $x_t$, while  problems $(\hat{\cP}_t)$ and  $(\bar{\cP}_t)$ do not. That is to say, problems  $(\hat{\cP}_t)$ and  $(\bar{\cP}_t)$ can be solved off-line once we are given the specific description of stochastic processes of $\{A_t,\B_t\}|_{t=0}^{T-1}$. Furthermore, these two convex optimization problems can be solved efficiently by existing modern numerical methods. The following result illustrates the relationship among problems  $({\cP}_t)$, $(\hat{\cP}_t)$ and  $(\bar{\cP}_t)$, which plays an important role in developing the explicit solution for problem $(\cP_{\LQ}^T)$.

\begin{theorem}\label{thm_sep}
For any $x$ $\in$ $\bR$, the optimal solution for problem $(\cP_t)$ is
\begin{align}
\u^*(x)=\begin{dcases}
\hK x & \textrm{if}~~x\geq 0,\\
\bK x & \textrm{if}~~x<0,
\end{dcases}\label{thm_sep_ut}
\end{align}
where $\hK$ and $\bK$ are defined respectively as
\begin{align*}
\hK&=\arg~\min_{\K \in \cK_t}~\hg_t(\K, y, z),\\
\bK&=\arg~\min_{\K \in \cK_t}~\bg_t(\K, y, z),
\end{align*}
and the optimal objective value is
\begin{align}
v(\cP_t)=x^2 \big(\hg_t(\hK,y,z) \1_{ \{ x \geq 0 \} }+\bg_t(\bK,y,z) \1_{ \{ x <0 \} } \big).\label{thm_sep_obj}
\end{align}	
\end{theorem}

\begin{IEEEproof}
Since problem $(\cP_t)$ is convex, the first-order optimality condition is sufficient to determine the optimal solution (see, e.g., Theorem 27.4 in \cite{Rockafellar:1970}). If $\u^*$ is the optimal solution, it should satisfy
\begin{align}
\nabla_{\u}g_t(\u^*,x,y,z)^{\prime}(\u-\u^*) \geq 0, ~\textrm{for}~\forall~\u \in \cU_t, \label{thm_sep_opt_ut}
\end{align}
where $\nabla_{\u}{g}_t(\u,x,y,z)$ is given in (\ref{def_nabla_gt}). Note that the condition (\ref{thm_sep_opt_ut}) depends on state $x$. Thus, we consider the following three different cases.

(i) We first consider the case of $x$ $>$ $0$. Let $\hK$ be the optimal solution of problem $(\hat{\cP}_t)$, which satisfies the following first-order optimality condition,
\begin{align}
\nabla_{\K}\hg_t(\hK, y, z)^{\prime}(\K-\hK) \geq 0, ~\textrm{for}~\forall~ \K \in \cK_t, \label{thm_sep_hatKt}
\end{align}
where $\nabla_{\K}\hg_t(\hK)$ is defined as
\begin{align}
&\nabla_{\K}\hg_t(\hK,y,z)=2\E_t \Big[(\R_t\hK+\S_t)
	+ \B_t^{\prime}(A_t + \B_t\hK) \notag \\
	&~~~~~\times\big(y\1_{\{A_t + \B_t\hK \geq 0\}} + z\1_{\{A_t+ \B_t\hK <0\}}\big)\Big].
\end{align}
If we let $\u^*=x \hK$, it is not hard to verify that $\u^*$ satisfies both the constraint $\u^* \in \cU_t$ and  the first-order optimality condition of $(\cP_t)$ by substituting $\u^*$ back to (\ref{thm_sep_opt_ut}) and using the condition (\ref{thm_sep_hatKt}). That is to say, $\u^*$ solves problem $(\cP_t)$ when $x>0$. Substituting $\u^*$ back to $(\ref{def_gt})$ gives the optimal objective value of ($\cP_t$) as $g_t(\u^*,x,y,z)$ $=$ $x^2 \hg_t(\hK,y,z)$.

(ii) For the case of $x<0$, we consider the first-order optimality  condition of problem $(\bg_t)$,
\begin{align*}
\nabla_{\K}\bg_t(\bK,y,z)^{\prime}(\K_t - \bK_t) \geq 0, \textrm{for}~\forall~ \K_t \in \cK_t,
\end{align*}
where $\nabla_{\K}\bg_t(\bK,y,z)$ is defined as,
\begin{align*}
&\nabla_{\K}\bg_t(\bK, y, z)=2\E_t \Big[(\R_t\bK+\S_t)
+ \B_t^{\prime}(A_t - \B_t \bK) \notag \\
&~~~~~\times\big(y\1_{\{ A_t - \B_t^{\prime}\bK \leq 0\}}
+z\1_{\{A_t- \B_t\bK >0\}}\big)\Big].
\end{align*}
Similarly, let $\u^*=-x \bK$. We can verify that $\u^*$ satisfies both the constraint $\u^*\in \cU_t$  and  the optimality condition (\ref{thm_sep_opt_ut}) of problem $(\cP_t)$. Thus, $\u^*$ solves problem $(\cP_t)$ for $x<0$ with the optimal objective value being $g_t(\u^*)$ $=$ $x^2 \bg_t(\bK)$.

(iii) When $x$ $=$ $0$, the objective function (\ref{def_gt}) of $(\cP_t)$ becomes
\begin{align}
g_t(\u,0,y,z)&=\E_t\Big[\u^{\prime}\R_t\u + \u^{\prime}\B_t^{\prime}\B_t\u\Big(y\1_{\{\B_t\u\geq 0\}}\notag\\
&+z\1_{\{\B_t\u< 0\}}\Big) \Big]. \label{thm_sep_eq1}
\end{align}
From the inequalities in (\ref{lem_convex_eq2}), we have
\begin{align*}
\R_t+\big(y\B_t^{\prime}\B_t \1_{\{\B_t\u\geq 0\}}+z\B_t^{\prime}\B_t\1_{\{\B_t\u< 0\}}\big) \succeq \M_t,
\end{align*}
where $\M_t$ is defined in (\ref{def_Mt}). Note that (\ref{thm_sep_eq1}) is bounded from below by
\begin{align*}
g_t(\u,x,y,z)|_{x=0}\geq \u^{\prime} \E_t[\M_t]\u \geq 0.
\end{align*}
As we have showed $\E_t[\M_t]\succ 0$, $g_t(\u^*,0,y,z)$ $=$ $0$ only when $\u^*=\0_{n\times 1}$. Clearly, $\u^*=\0_{n\times 1}$ also satisfies the constraint $\cU_t$. That is to say, $\u^*=\0_{n\times 1}$ solves problem $(\cP_t)$ when $x=0$. As a summary of the above three cases, the optimal solution of problem ($\cP_t$) can be expressed as (\ref{thm_sep_ut}) and the optimal objective value is given as (\ref{thm_sep_obj}).
\end{IEEEproof}

\subsection{Explicit solution for problem ($\cP_{\LQ}^T$)}\label{sse_sol}
With the help of Theorem \ref{thm_sep}, we can develop the explicit solution for problem $(\cP_{\LQ}^T)$. We first introduce the following two random sequences, $\hG_0$, $\hG_1$, $\cdots$, $\hG_T$ and $\bG_0$, $\bG_1$, $\cdots$, $\bG_T$, which are defined backward recursively as follows,
\begin{align}
\hG_t&:=\min_{\K_t \in \cK_t}~\hg_t(\K_t, \hG_{t+1},\bG_{t+1}), \label{def_hatGt}\\
\bG_t&:=\min_{\K_t \in \cK_t}~\bg_t(\K_t, \hG_{t+1},\bG_{t+1}), \label{def_barGt}
\end{align}
where $\hg_t(\cdots)$ and $\bg_t(\cdots)$ are defined respectively in (\ref{def_hgt}) and (\ref{def_bgt}) for $t= T-1,\cdots, 0$ with the boundary conditions of $\hG_T$ $=$ $\bG_T$ $=$ $q_T$. Clearly, $\hG_t$ and $\bG_t$ are $\cF_t$-measurable random variables.

\begin{theorem} \label{thm_PLQ}
The optimal control policy of problem $(\cP_{\LQ}^T)$ at time $t$ is a linear feedback policy£¬
\begin{align}
u_t^*(x_t)=\begin{dcases}
 \hK^*_t x_t &~\textrm{if}~x_t\geq 0,\\
- \bK^*_t x_t &~\textrm{if}~x_t <0,\\
\end{dcases}\label{thm_PLQ_ut}
\end{align}
where $\hK^*_t$ and $\bK^*_t$ are defined as,
\begin{align*}
\hK^*_t&=\arg \min_{\K_t\in \cK_t}~ \hg_t(\K_t, \hat{G}_{t+1}, \bar{G}_{t+1}),\\
\bK^*_t&=\arg \min_{\K_t\in \cK_t}~\bg_t(\K_t, \hat{G}_{t+1}, \bar{G}_{t+1}),
\end{align*}
where $\{\hG_t\}|_{t=0}^T$ and $\{\bG_t\}|_{t=0}^T$ are given in (\ref{def_hatGt}) and (\ref{def_barGt}), respectively, and $\hG_t>0$ and $\bG_t>0$ for $t=0,\cdots, T-1$. Furthermore, the optimal objective value of problem $(\cP_{\LQ}^T)$ is
\begin{align}
v(\cP_{\LQ}^T)=x_0^2\left(\hG_0 \1_{\{x_0\geq 0\}} + \bG_0 \1_{\{x_0<0 \}} \right).
\end{align}
\end{theorem}

\begin{IEEEproof} We prove this theorem by invoking dynamic programming. At any time $t=0,\cdots, T$, the value function of problem $(\cP_{\LQ}^T)$ is defined as
\begin{align*}
V_t(x_t)&:=\min_{\u_t,\u_{t+1},\cdots,\u_{T-1}} \E_t\left[ \sum_{k=t}^{T} \left(\begin{array}{cc}
\u_k\\
x_k
\end{array}\right)^{\prime}\C_k
\left(\begin{array}{cc}
\u_k\\
x_k
\end{array}\right) \right],\\
{\rm s.t.}~&\{x_k,\u_k\}~\textrm{satisfies (\ref{def_xt}) and (\ref{def_Ut}) for $k=t,\cdots,T-1$}.
\end{align*}
From the Bellmen's principle of optimality, the value function $V_t(x_t)$ satisfies the reversion,
\begin{align}
V_t(x_t)&=\min_{\u_t \in \cU_t}
\left(\begin{array}{cc}
\u_t\\
x_t
\end{array}\right)^{\prime}C_t
\left(\begin{array}{cc}
\u_t\\
x_t
\end{array}\right)\notag\\
&~~~~~~~~~~~~~~~~+\E_t[V_{t+1}(x_{t+1})]. \label{thm_PLQ_Vt}
\end{align}
By using the mathematical induction, we show that the following claim is true,
\begin{align}
V_t(x_t)=x_t^2 (\hG_{t}\1_{\{x_t\geq 0\}} +\bG_t\1_{\{x_t<0\}}) \label{thm_PLQ_claim}
\end{align}
for $t=T,T-1,\cdots,1$, where $\hG_t$ and $\bG_t$ satisfy the recursions, respectively, in (\ref{def_hatGt}) and (\ref{def_barGt}). Clearly, at time $t=T$, the value function is
\begin{align*}
V_T(x_T)= q_T x_T^2.
\end{align*}
As we have defined $\hG_T$ $=$ $\bG_T$ $=$ $q_T$, the claim (\ref{thm_PLQ_claim}) is true. Now, we assume that the claim (\ref{thm_PLQ_claim}) is true at time $t+1$,
\begin{align*}
V_{t+1}(x_{t+1})=x_{t+1}^2 (\hG_{t+1}\1_{\{x_{t+1}\geq 0\}} +\bG_{t+1}\1_{\{x_{t+1}<0\}}).
\end{align*}
From (\ref{thm_PLQ_Vt}) and the definition of $g_t$ in (\ref{def_gt}) at time $t$, we have
\begin{align*}
V_t(x_t)=\min_{\u_t \in \cU_t}~g_t(\u_t,x_t, \hG_{t+1}, \bG_{t+1}),
\end{align*}
which is just the same problem as ($\cP_t$) given in (\ref{def_Pt}) by replacing $y$ and $z$ with $\hG_{t+1}$ and $\bG_{t+1}$, respectively. From Theorem \ref{thm_sep}, we have the result in (\ref{thm_PLQ_ut}) and the claim in (\ref{thm_PLQ_claim}) is true at time $t$, by defining $\hG_t=\hg_t(\hK_t^*, \hG_{t+1}, \bG_{t+1} )$ and $\bG_t=\bg_t(\bK_t^*, \hG_{t+1}, \bG_{t+1} )$, which completes the proof.
\end{IEEEproof}

The above theorem indicates that the system of equations (\ref{def_hatGt}) and (\ref{def_barGt}) play the same role as the Riccati Equation for the classical LQG optimal problem. In the following part, we name the pair of (\ref{def_hatGt}) and (\ref{def_barGt}) as the \textit{Extended Riccati Equations}. Furthermore, when there is no control constraint, these two equations will merge to the one equation in  (\ref{def_GGt}) presented later in Section \ref{sse_noconstraint}.

In Theorem \ref{thm_PLQ}, the key step in identifying $\u_t^*$ is to solve (\ref{def_hatGt}) and (\ref{def_barGt}) for $\hK^*_t$ and $\bK^*_t$, respectively, for each $t=T-1,\cdots, 0$. Generally speaking, once the stochastic process of $\{A_t\}|_{t=0}^{T-1}$ and $\{\B_t\}|_{t=0}^{T-1}$ is specified, we can employ the following procedure to compute $\bK^*_t$ and $\hK^*_t$:

\begin{algorithm}[htb]
\caption{Computation of optimal gains of $\hK_t^*$ and $\bK_t^*$}
\label{alg:hKbK}
\begin{algorithmic}[1]
\State Let $\hG_T$ $\leftarrow $ $q_T$, $\bG_T$ $\leftarrow$ $q_T$. \label{code:step 1}
\For{$k \gets T-1, T-2, \cdots,0 $}
\State Solve problems $(\ref{def_hatGt})$ and $(\ref{def_barGt})$ for $\hK^*_k$ and $\bK^*_k$;
\State Compute
\begin{align*}
\hG_k \leftarrow \hg_k(\hK_k,\hG_{k+1}, \bG_{k+1})\\
\bG_k \leftarrow \bg_k(\bK_k,\hG_{k+1}, \bG_{k+1})
\end{align*}
\EndFor
\end{algorithmic}
\end{algorithm}
Note that, if $\{A_t\}|_{t=0}^{T-1}$ and $\{\B_t\}|_{t=0}^{T-1}$ are statistically independent over time, all the conditional expectations in $\hg_t(\cdots)$ and $\bg_t(\cdots)$ would degenerate to the unconditional expectations, which generates the deterministic pairs of $\{$$\hG_t$, $\bG_t$ $\}$ and $\{$ $\hK^*_t$, $\bK^*_t$ $\}$. However, if $\{A_t\}|_{t=0}^{T-1}$ or $\{\B_t\}|_{t=0}^{T-1}$ are serially correlated over time, all the conditional expectations $\E_t[\cdot]$ depend on the filtration $\cF_t$, or in other words, all these pairs are also random variables. Under such a case, usually, we need numerical methods to discreterize the sample space and solve both problems, (\ref{def_hatGt}) and (\ref{def_barGt}), for each sample path.

\subsection{Solution for problem ($\cP_{\LQ}^T$) with no control constraints}\label{sse_noconstraint}
If there is no control constraint in $(\cP_{\LQ}^T)$, i.e., $\H_t=\0$, Theorem \ref{thm_PLQ} can be simplified to: for all $t=$ $T-1$, $\cdots$, $0$, solving problems $(\hat{\cP}_t)$ and $(\bar{\cP}_t)$ yields $\hG_t$ $=$ $\bG_t$ and $\hK_t^*$$=$$-\bK^*_t$, respectively. Let $\K^*_t$$:= $$\hK^*_t$  and $G_t=\hG_t$, for $t=$$T-1$, $\cdots$, $0$. Thus, we have the following explicit control policy.
\begin{proposition}\label{prop_PT_nocont}
When there is no control constraint in  problem ($\cP_{\LQ}^T$),  the optimal control becomes
\begin{align*}
\u_t^*=-\big(\R_t+\E_t[G_{t+1}\B_t^{\prime}\B_t ]\big)^{-1}( \E_t[G_{t+1}A_t\B_t^{\prime}] +\S_t)x_t,
\end{align*}
for $t=T-1, \cdots,0$, where $G_{t}$ is defined by
\begin{align}
G_t&=q_t +\E_t[G_{t+1}A_t^2]-(\S_t +\E_t[G_{t+1}A_t\B_t^{\prime}])^{\prime}\notag\\
   &~\times (\R_t +\E_t[G_{t+1}\B_t^{\prime}\B_t ])^{-1}(\S_t +\E_t[G_{t+1}A_t\B_t^{\prime}]),\label{def_GGt}
\end{align}
with $G_T=q_T$.
\end{proposition}
Proposition \ref{prop_PT_nocont} is a generalization of Corollary 5 in \cite{GaoLiCuiWang:2015}, which solves the dynamic mean-variance portfolio control problem with correlated returns.

\section{Optimal Solution to Problem $(\cP_{\LQ}^{\infty})$}\label{se_sol_Pinf}
In this section, we develop the optimal solution for $(\cP_{\LQ}^{\infty})$. The main idea is to study the asymptotic behavior of the solution from the correspondent finite-horizon problem by extending the time horizon  to infinity. More specifically, we consider the following problem with finite  horizon,
\begin{align*}
(\cA^T )~&~\min_{\{\u_t\}|_{t=0}^{T-1}}~\E\left[ ~~\sum_{t=0}^{T-1}
\left(\begin{array}{cc}
\u_t\\
x_t
\end{array}\right)^{\prime}\C
\left(\begin{array}{cc}
\u_t\\
x_t
\end{array}\right) \right]\\
{\rm s.t.}~&~\{x_t,\u_t\} ~~\textrm{satisfies (\ref{def_xt_inf}) and (\ref{def_Ut_inf}) for $t=0,\cdots,T-1$},
\end{align*}
where $x_0$ is given. Obviously, $(\cA^{T})$ becomes $(\cP_{\LQ}^{\infty})$ when $T$ goes to infinity.

Problem $(\cA^T)$ can be regarded as a special case of problem $(\cP_{\LQ}^T)$ studied in Section \ref{se_sol_P}, when we assume $\{A_t,\B_t\}$ to be i.i.d and $q_T=0$. Theorem \ref{thm_PLQ} can be thus applied to solve $(\cA^T)$. In order to study $(\cA^T)$ for different time horizon $T$, we use the notations $\{$$\hG_0^T$, $\hG_1^T$,$\cdots$, $\hG_T^T$$\}$ and $\{$$\bG_0^T$, $\bG_1^T$,$\cdots$, $\bG_T^T$$\}$ to denote the outputs of recursions (\ref{def_hatGt}) and (\ref{def_barGt}) with the boundary condition $\hG_T^T$ $=$ $\bG_T^T$ $=0$, respectively. From Theorem \ref{thm_PLQ}, the optimal value of problem $(\cA^T)$ can be expressed as
\begin{align}
v(\cA^T)=x_0^2\Big( \hG_0^T \1_{\{x_0\geq 0\}} + \bG_0^T \1_{\{x_0< 0\}} \Big). \label{def_v(cP)}
\end{align}
 As $x_0$ is known, the optimal value $v(\cA^T)$ is either $x_0^2 \hG^T_0$ or $x_0^2 \bG_0^T$. The output sequences, $\{\hG_t^T\}|_{t=0}^T$ and $\{\bG_t^T\}|_{t=0}^T$, have the following property.
\begin{lemma}\label{lem_G_homo}
Consider problems $(\cA^T)$ and $(\cA^{T+k})$ for some $k\in \{1,2,\cdots\}$. If $\{$$\hG^T_0$,$\cdots$,$\hG^T_{T}$$\}$, $\{$$\bG^T_0$,$\cdots$,$\bG^T_{T}$$\}$  and $\{$$\hG^{T+k}_0$,$\cdots$,$\hG^{T+k}_{T+k}$$\}$, $\{$$\bG^{T+k}_0$,$\cdots$,$\bG^{T+k}_{T+k}$$\}$ are the outputs of recursions (\ref{def_hatGt}) and (\ref{def_barGt}) of these two problems, respectively, then $\hG^T_t$$=$$\hG^{T+k}_{t+k}$, $\bG^T_t$$=$$\bG^{T+k}_{t+k}$
for all $t=0,1,\cdots, T-1$.
\end{lemma}
The proof of Lemma \ref{lem_G_homo} is straightforward by applying Theorem \ref{thm_PLQ} respectively to problems  $(\cA^T)$ and $(\cA^{T+k})$, and noticing the same boundary conditions $\hG^{T}_{T}$$=$$\hG^{T+k}_{T+k}$$=$$0$ and
$\bG^{T}_{T}$$=$$\bG^{T+k}_{T+k}$$=$$0$ for the recursions (\ref{def_hatGt}) and (\ref{def_barGt}).  Lemma \ref{lem_G_homo} basically shows that the solution sequences $\{\hG^T_t\}|_{t=0}^T$ and $\{\bG^T_t\}|_{t=0}^T$  are time-homogenous, i.e., the values of these two sequences do not depend on the exact time $t$ but depend on the total recursion numbers in (\ref{def_hatGt}) and (\ref{def_barGt}).

Let $\mathcal{K}$$:=$$\{$$\K$$\in$$\bR^n$$~|~$$\H\K$$\leq$$\d$$\}$. We present now the main result for problem $(\cP_{\LQ}^{\infty})$.
\begin{theorem}\label{thm_PLQinf}
For problem $(\cA^T)$, if there exists some $M>0$, such that $\hG_0^T<M$ and $\bG_0^T<M$   for any $T$,  then the following system of equations admits a pair of solution, $\hG^* >0$ and $\bG^*>0$,
\begin{align}
\hG^* &:=\min_{\K \in \cK}\hg(\K, \hG^*,\bG^*),\label{def_equ_hG}\\
\bG^* &:=\min_{\K \in \cK}\bg(\K, \hG^*,\bG^*),\label{def_equ_bG}
\end{align}
where $\hg(\K, y, z)$ $:$ $\bR^m$ $\times$ $\bR_+$ $\times$ $\bR_+$ $\rightarrow$ $\bR_+$ and $\bg(\K, y, z)$ $:$ $\bR^m$ $\times$ $\bR_+$ $\times$ $\bR_+$ $\rightarrow$ $\bR_+$  are defined as
\begin{align}
&\hg(\K, y, z )=\E \Big[ \left(\begin{array}{cc}
		\K\\
		1
	\end{array}\right)^{\prime}\C
	\left(\begin{array}{cc}
		\K\\
		1
	\end{array}\right)+(A+\B\K)^2\notag\\
	&~~~\times\big(y \1_{\{A+\B\K\geq 0\}} + z\1_{\{A +\B\K <0\}}\big)\Big],\label{def_hatg}\\
&\bg(\K,y, z)=\E \Big[ \left(\begin{array}{cc}
		-\K\\
		1
	\end{array}\right)^{\prime}\C
	\left(\begin{array}{cc}
		-\K\\
		1
	\end{array}\right)+(A-\B\K)^2\notag\\
&~~~\times\big(y\1_{\{A-\B\K \leq 0\}} +z\1_{\{A_t-\B_t\K >0\}}\big)\Big],\label{def_barg}
\end{align}
and the policy
\begin{align}
\u_t^*=\hK^*x_t \1_{\{x_t\geq 0\}} +\bK^* x_t \1_{\{x_t<0\}},\label{thm_PLQinf_ut}
\end{align}
$t=0,\cdots, \infty$, solves problem $(\cP_{\LQ}^{\infty})$, where
\begin{align}
\hK^*&=\arg \min_{\K \in \cK} \hg(\K, \hG^*, \bG^* ), \label{thm_PLQinf_hK}\\
\bK^*&=\arg \min_{\K \in \cK}  \bg(\K, \hG^*,\bG^* ). \label{thm_PLQinf_bK}
\end{align}
Furthermore, under the optimal control $u_t^*$ the closed-loop system,
$x_{t+1}$ $=$ $x_t\big( A+\B(\hK^* \1_{\{x_t\geq 0\}}$ $+$ $\bK^*\1_{\{x_t<0\}})\big)$,
is $L^2$-asymptotically stable, i.e., $\lim_{t\rightarrow \infty}$$\E[(x_t^*)^2]$$=$$0$.
\end{theorem}

\begin{IEEEproof}
We first show that the optimal value of problem $(\cA^T)$ is nondecreasing when $T$ increases. Suppose that the optimal control policy $ \{\tilde{\u}_t\}|_{t=0}^{T}$ solves problem $(\cA^{T+1})$. The following is evident,
\begin{align}
&v(\cA^{T+1})=\E \left[ \sum_{t=0}^{T}
\left(\begin{array}{cc}
\tilde{\u}_t\\
x_t
\end{array}\right)^{\prime}\C
\left(\begin{array}{cc}
\tilde{\u}_t\\
x_t
\end{array}\right) \right] \notag\\
&~~\geq \E \left[ \sum_{t=0}^{T-1}
\left(\begin{array}{cc}
\tilde{\u}_t\\
x_t
\end{array}\right)^{\prime}\C
\left(\begin{array}{cc}
\tilde{\u}_t\\
x_t
\end{array}\right) \right]\geq v(\cA^{T}). \label{thm_PLQinf_ineq}
\end{align}
Applying both (\ref{thm_PLQinf_ineq}) and (\ref{def_v(cP)}) gives rise to $\hG^{T+1}_0\geq \hG^T_0$ and $ \bG^{T+1}_0\geq \bG^T_0$. Thus, $\hG^{T}_0$ and $\bG^{T}_0$ are nondecreasing sequences with respect to $T$. If there exists $M$ such that $\hG^{T}_0<M$  and $\bG^{T}_0<M$ for any $T$, both limits of $\hG^{T}_0$ and $\bG^{T}_0$ exist. We denote $\hG^{\infty}_0$ $:=$ $\lim_{T\rightarrow \infty}$ $\hG^T_0$ and $\bG^{\infty}_0$ $:=$ $\lim_{T\rightarrow \infty}$ $\bG^T_0$. Now, we derive the relationships between $\hG^T_0$ and $\hG^{T+1}_0$, and between $\bG_0^T$  and $\bG_0^{T+1}$, respectively. The value function $V_1(x_1)$ of problem $(\cA_{T+1})$ is given as follows from (\ref{thm_PLQ_claim}),
\begin{align}
V_1(x_1)&=x_1^2 \big( \hG_1^{T+1}\1_{\{x_1\geq 0\}} +\bG^{T+1}_1\1_{\{x_1< 0\}} \big)\notag\\
&=x_1^2 \big( \hG_0^{T}\1_{\{x_1\geq 0\}} +\bG^{T}_0\1_{\{x_1< 0\}} ) ,\label{thm_PLQinf_V1}
\end{align}
where the last equality is from Lemma \ref{lem_G_homo}. Then at time $t=0$, the recursion (\ref{thm_PLQ_Vt}) implies
\begin{align}
V_0(x_0)=\min_{\u_0 \in \cU(x_0) } \E\left[ \left(\begin{array}{cc}
\u_t\\
x_t
\end{array}\right)^{\prime}\C
\left(\begin{array}{cc}
\u_t\\
x_t
\end{array}\right)+ V_1(x_1)\right].\label{thm_PLQinf_V0}
\end{align}
Substituting (\ref{thm_PLQinf_V1}) to (\ref{thm_PLQinf_V0}) rewrites the objective function of (\ref{thm_PLQinf_V0}) as $g_t(\u_0,x_0,\hG^{T}_{0},\bG^{T}_{0} )$ defined in (\ref{def_gt}). Thus, applying Theorem \ref{thm_sep} yields the optimal value of problem (\ref{thm_PLQinf_V0}) as
\begin{align*}
V_0(x_0)=x_0^2 ( \hG_0^{T+1}\1_{\{x_0\geq 0\} }+\bG_0^{T+1}\1_{\{x_0< 0\}} ),
\end{align*}
where  $\hG^{T+1}_0$ and $\bG^{T+1}_0$ are given by
\begin{align}
\hG^{T+1}_0 &=\min_{\K\in\cK } \hg(\K, \hG_0^T, \bG_0^T), \label{thm_PLQinf_hGT}\\
\bG^{T+1}_0 &=\min_{\K\in\cK } \bg(\K, \hG_0^T, \bG_0^T). \label{thm_PLQinf_bGT}
\end{align}
When $T$ goes to $\infty$, if both the limits of $\hG^T_0$ and $\bG^T_0$ exist, the equations (\ref{thm_PLQinf_hGT}) and (\ref{thm_PLQinf_bGT}) will converge to the equations (\ref{def_equ_hG}) and (\ref{def_equ_bG}), respectively.  Under such a case, the stationary optimal policy (\ref{thm_PLQinf_ut}) becomes optimal, where the correspondent $\hK^*$ and $\bK^*$ are the minimizers of (\ref{def_equ_hG}) and (\ref{def_equ_bG}), respectively.

We show next that under the optimal control $\u_t^*$, the closed-loop system (\ref{thm_PLQinf}) is asymptotically stable. At any time $t$, we first assume $x_t^*\geq 0$. Then, $x_{t+1}^*$ $=$ $(A+\B\hK^*)x_t^*$ and the following equality holds,
\begin{align}
&\E\big[(x_{t+1}^*)^2(\hG^*\1_{\{x_{t+1}^*\geq 0\}} +\bG^* \1_{\{x_{t+1}^*<0\}})\big]-(x_t^*)^2 \hG^* \notag\\
&=\E\big[(x_t^*)^2(A+\B\hK^*)^2\big(\hG^*\1_{\{A+\B\hK^* \geq0\}}\notag\\
&+\bG^* \1_{\{A+\B\hK^* <0\}}\big)\big]-(x_t^*)^2\hG^*\notag\\
&=-(x_t^*)^2\left(
             \begin{array}{c}
               \hK^* \\
               1 \\
             \end{array}
           \right)^{\prime}\C \left(
                                \begin{array}{c}
                                  \hK^* \\
                                  1 \\
                                \end{array}
                              \right),\label{thm_PLQinf_eq1}
\end{align}
where the last equality is from (\ref{def_equ_hG}). When $x_t<0$, we can derive a similar equality,
\begin{align}
&\E[(x_{t+1}^*)^2(\hG^*\1_{\{x_{t+1}^*\geq 0\}} +\bG^* \1_{ \{x_{t+1}^*<0\}})]-(x_t^*)^2 \bG^* \notag\\
&=-(x_t^*)^2\left(
             \begin{array}{c}
               -\bK^* \\
               1 \\
             \end{array}
           \right)^{\prime}\C \left(
                                \begin{array}{c}
                                  -\bK^* \\
                                  1 \\
                                \end{array}
                              \right).\label{thm_PLQinf_eq2}
\end{align}
Combining (\ref{thm_PLQinf_eq1}) and (\ref{thm_PLQinf_eq2}) gives rise to
\begin{align}
&\E\big[(x_{t+1}^*)^2(\hG^*\1_{\{x_{t+1}^*\geq 0\}} +\bG^* \1_{\{x_{t+1}^*<0\}})\big]\notag\\
&~~-(x_t^*)^2 (\hG^*\1_{\{x_t^*\geq 0\}}+ \bG^*\1_{\{x_t^*<0\}} ) \notag\\
&=-(x_t^*)^2\big( \hat{J}\1_{\{x_t^*\geq 0\}}+ \bar{J}\1_{\{x_t^*<0\}} \big),\label{thm_PLQinf_eq3}
\end{align}
where
\begin{align*}
\hat{J}=\left(
             \begin{array}{c}
               \hK^* \\
               1 \\
             \end{array}
           \right)^{\prime}\C\left(
                                \begin{array}{c}
                                  \hK^* \\
                                  1 \\
                                \end{array}
                              \right),~~\bar{J}= \left(
             \begin{array}{c}
               -\bK^* \\
               1 \\
             \end{array}
           \right)^{\prime}\C\left(
                                \begin{array}{c}
                                  -\bK^* \\
                                  1 \\
                                \end{array}
                              \right).
\end{align*}
%
%
Then, applying (\ref{thm_PLQinf_eq3}) successively $t$ times and taking the expectation yields,
\begin{align}
&\E[(x_{t+1}^*)^2(\hG^*\1_{\{x_{t+1}^*\geq 0\}} +\bG^* \1_{\{x_{t+1}^*<0\}})]\notag\\
&=x_0^2(\hG^*\1_{\{x_0\geq 0\}}+\bG^*\1_{\{x_0<0\}} )\notag\\
&~~-\sum_{k=0}^{t}\E\left[(x_k^*)^2\big(\hat{J}\1_{\{x_k^*\geq 0\}}+\bar{J} \1_{\{x_k^*<0\}}\big)\right].\label{thm_PLQinf_eq4}
\end{align}
Since $\hG^*\geq 0$ and $\bG^* \geq 0$, the left hand side of (\ref{thm_PLQinf_eq4}) is nonnegative for all $t$. Thus,
\begin{align}
\lim_{t\rightarrow \infty}\E\left[(x_t^*)^2\big(\hat{J}\1_{\{x_t^*\geq 0\}}+\bar{J} \1_{\{x_t^*<0\}}\big)\right]=0. \label{thm_PLQinf_eq5}
\end{align}
From Assumption \ref{asmp_C},  we know $\hat{J}>0$ and $\bar{J}>0$. Thus, (\ref{thm_PLQinf_eq5}) implies $\lim_{t\rightarrow \infty} \E[(x^*_t)^2]=0$.

\end{IEEEproof}

Although Theorem \ref{thm_PLQinf} characterizes the solution of $(\cP_{\LQ}^{\infty})$, the existence condition for the solutions to (\ref{def_equ_hG}) and (\ref{def_equ_bG}) is not easy to check. Motivated by Proposition 4.17 in \cite{Bertsekas:2005}, we can adopt the following algorithm to check whether the solutions exist for (\ref{def_equ_hG}) and (\ref{def_equ_bG}) in real applications.

\begin{algorithm}[htb]
\caption{Iterative solution method for (\ref{def_equ_hG}) and (\ref{def_equ_bG}) }
\label{alg:hGbG}
\begin{algorithmic}[1]
\Require
The information of the distribution of $A$ and $\B$, the threshold parameter $\epsilon>0$ and the maximum number of the iteration, $I_{\max}$.
\Ensure
\State Let $i\leftarrow 0$, $\hG_i\leftarrow 0$ and $\bG_i \leftarrow0$. \label{code:step 1}
\State  Calculate
  \begin{align}
  &\hG_{i+1} \leftarrow \min_{\K\in  \cK} \hg(\K,\hG_{i},\bG_{i}), \label{alg:hGbG_hG}\\
  &\bG_{i+1} \leftarrow \min_{\K \in \cK} \bg(\K, \hG_{i},\bG_{i}).\label{alg:hGbG_bG}
  \end{align}
\label{code:step 2}
\State If $|\hG_{i+1}-\hG_{i}|$ $\leq$ $\epsilon$ and $|\bG_{i+1}-\bG_{i}|$ $\leq$ $\epsilon$, go to Step \ref{code:return_suc};
\State If $i>I_{\max}$, go to Step \ref{code:return_fail};
\State Let $i\leftarrow i+1$ and go to Step \ref{code:step 2};
\label{code:fram:classify}\\
\Return $\hG^*\leftarrow \hG_{i+1}$, $\bG^*\leftarrow \bG_{i+1}$, and $\hK^*$ and $\bK^*$ as calculated in (\ref{thm_PLQinf_hK}) and (\ref{thm_PLQinf_bK}), respectively; \label{code:return_suc}
\State Equations (\ref{def_equ_hG}) and (\ref{def_equ_bG}) do not have solution. \label{code:return_fail}
\end{algorithmic}
\end{algorithm}

%

For problem ($\cP_{\LQ}^{\infty}$), due to the control constraint, we need to use an iterative process in Algorithm \ref{alg:hGbG} to check whether (\ref{def_equ_hG}) and (\ref{def_equ_bG}) admit a solution. However, it is possible to derive some sufficient conditions to guarantee the existence of the solution. Before we present our result, it is worth to mention some classical results. For the classical LQG problem\footnote{In the classical LQG model, the noises are additive in the state equation and there is no control constraints.}, it is well known that both the controllability and observability of the system guarantee the existence of the solution to the correspondent algebraic Riccati equation. However, when there is uncertainty in the system matrices $A$ and $B$\footnote{This kind of uncertain systems is also known as the system with multiplicative noise.}, it is well known that the controllability and observability no longer can ensure the existence of the solution for the correspondent algebraic Riccati equation (e.g., see \cite{AthanKuGershwin:1977} \cite{Yaz:1990}). Athans et al. \cite{AthanKuGershwin:1977} provide the condition called \textit{the uncertainty threshold principle}, which guarantees the existence of the solution of the algebraic Riccati equation resulted from the simple uncertain system with a scalar state and a scalar control. Roughly speaking, if a threshold computed by the systems parameters is less or equal than $1$, the correspondent algebraic Riccati equation admits a solution. The result given in \cite{Yaz:1990} further generalizes such a result to systems with a vector state and a vector control. However, because of the control constraint in $(\cP_{\LQ}^{\infty})$, the result of the uncertainty threshold principle does not hold for our problem. Let us use the following counter example to explain.
\begin{example}\label{exam_counter}
We consider an example of problem ($\cP_{\LQ}^{\infty}$) with $n$$=$$1$, $q$$=$$1$, $\R$$=$$1$ and $\S$ $=$ $0.1$. The uncertain system matrices are assumed to take the following 5 scenario with identical probability $=$ $0.2$, $(A,\B)$ $\in$ $ \{$ $(-0.8,-0.7)$, $(-0.4, -0.6)$, $(0.2,0.4)$, $(0.6,0.8 )$, $(0.9,1)\}$.
It is not hard to compute the threshold given in \cite{AthanKuGershwin:1977} as,
\begin{align*}
\textrm{Threshold}=\E[A^2]-\frac{(\E[A]\E[B])^2 }{\E[B^2]}=0.4014<1.
\end{align*}
According to the results in Athans et al. \cite{AthanKuGershwin:1977}, the correspondent algebraic Riccati equation has a unique solution. However, if we set control constraint as $3|x_t|$  $\leq $ $\u_t$ $\leq$ $4 |x_t|$ for all $t=0,\cdots, T-1$, applying Algorithm \ref{alg:hGbG} concludes that $\hat{G}^*$ and $\bar{G}^*$ go to infinity as the number of iteration increases (see Sub-figure (b) in Figure \ref{fig_hGbG}). Sub-figures (a) in Figure \ref{fig_hGbG} plots the output $\hat{G}^*$  and $\bar{G}^*$ generated by Algorithm \ref{alg:hGbG} for  $2|x_t|$$\leq$$ \u_t$$\leq$$3|x_t|$, which shows the convergence of $\hat{G}^*$  and $\bar{G}^*$.

\begin{figure}
  \centering
  \includegraphics[width=250pt]{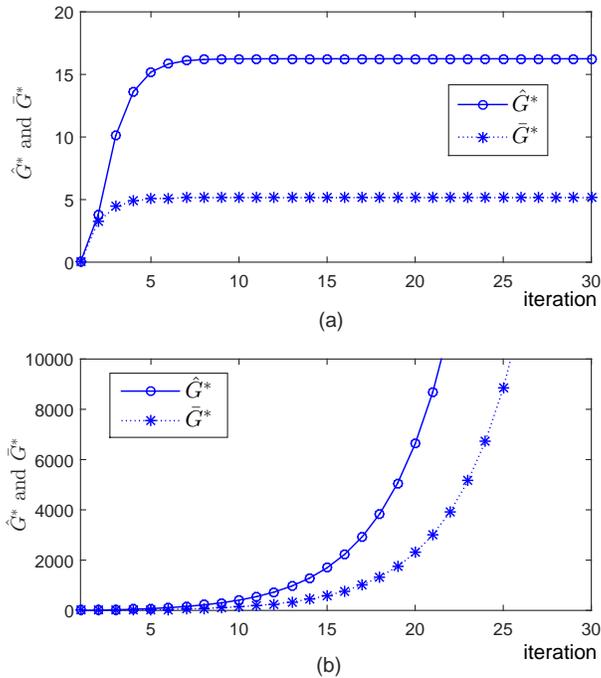}\\
  \caption{The outputs of $\hat{G}^*$ and $\bar{G}^*$ from the iterations of Algorithm \ref{alg:hGbG} for Example \ref{exam_counter}. The control constraint is $2|x_t|$$\leq$$ \u_t$$\leq$ $3|x_t|$ in Sub-figure (a) and $3|x_t|$$\leq$ $ \u_t$ $\leq $ $4|x_t|$ in sub-figure (b).}\label{fig_hGbG}
\end{figure}

\end{example}

From the previous example, we know that the conditions that guarantee the admissibility of (\ref{def_equ_hG}) and (\ref{def_equ_bG}) are quite complicated, since they depend on multiple parameters $A$, $\B$, $\C$, $\H$ and $\d$. Fortunately, we are able to derive the following sufficient conditions which are easily-checkable .

\begin{theorem}\label{thm_sufft}
For  problem $(\cP_{\LQ}^{\infty})$, if $\E[A^2]+\eta K_{\max} <1$, where $\eta$ is the maximum eigenvalue of $\E[\B^{\prime}\B]$ and $K_{\max}<\infty$ is a finite upper bound of $\|\K\|$ such that
$K_{\max}$ $=$ $\max$$\{$$\|\K \|$$~|~$$\H\K\leq \d\}$, then (\ref{def_equ_hG}) and (\ref{def_equ_bG}) admit a solution.
\end{theorem}
\begin{IEEEproof}
We consider one cycle of iteration of Algorithm \ref{alg:hGbG} given in (\ref{alg:hGbG_hG}) and (\ref{alg:hGbG_bG}). We first check the optimization problem in (\ref{alg:hGbG_hG}). Introducing the Lagrange multiplier $\lambda\in \R^m_+$ yields the following Lagrange function, $L(\K, \lambda):$ $=$ $ \hg(\K, \hG_i, \bG_i) +\lambda^{\prime}(\H \K-\d).$
From the convex duality theory,  the optimal solution can be characterized as
\begin{align}
\hK^*:=\arg \min_{\K \in \R^n }~ L(\K,\lambda^*) \label{thm_sufft_Klambda}
\end{align}
with $\lambda^*$ being the optimal solution of the dual problem,
$\lambda^*$$=$$\arg \max_{\lambda \geq 0}$$\{\min_{\K\in \R^n} L(\K,\lambda)\}$. The first order optimality condition of problem (\ref{thm_sufft_Klambda}) is given as
\begin{align}
&\nabla_{\K} L(\K,\lambda)=\E\big[ 2\R\hK^* +2\S +2\B^{\prime}(A+\B \hK^*)\times\notag\\
&(\hG_i\1_{\{A+\B\hK^*\geq 0 \}}+\bG_i\1_{\{A+\B\hK^* < 0 \}}) \big]+\H^{\prime}\lambda=\0_{n\times 1}. \label{thm_sufft_condition}
\end{align}
Substituting $\hK^*$ into (\ref{alg:hGbG_hG}) gives rise to
\small
\begin{align}
&\hG^*_{i+1}=\E\big[(\hK^*)^{\prime}\R \hK^* +2\S^{\prime}\hK^*+q+ \big(A^2+ 2A\B\hK^*\notag\\
&+(\hK^*)^{\prime}\B^{\prime}\B\hK^{*}\big)
\big(\hG_i\1_{\{A+\B\hK^*\geq 0\}}+\bG_i \1_{\{A+\B\hK^*< 0\}} \big).\label{thm_sufft_Gi+1}
\end{align}
\normalsize
Note that the complementary slackness condition is given as $\lambda^{\prime}$ $(\H\K-\d)$ $=$ $0$. Now, combining (\ref{thm_sufft_condition}) with (\ref{thm_sufft_Gi+1}) yields
\begin{align*}
&\hG_{i+1}=\E[ (A^2- (\hK^*)^{\prime}\B^{\prime}\B \hK^*)(\hG_i \1_{\{A+\B\hK^*\geq 0} \notag\\
&+\bG_i \1_{\{A+\B\hK^*< 0\}})]+q-(\hK^*)^{\prime}\R\hK^*-\d^{\prime}\lambda^*.\notag
\end{align*}
Note that, under our assumption, the last three terms are bounded by some constant $L$. Thus,
\begin{align}
&\hG_{i+1}\leq \big|\E[ A^2- (\hK^*)^{\prime}\B^{\prime}\B \hK^*]\big| \max\{\hG_i,\bG_i\}+ L\notag\\
&\leq (\E[A^2]+\eta K_{\max})\max\{\hG_i,\bG_i \} +L.\label{thm_sufft_ineq1}
\end{align}
Similarly, for $\bG_{i+1}$, we can derive the following inequality,
\begin{align}
\bG_{i+1}\leq (\E[A^2]+\eta K_{\max})\max\{\hG_i,\bG_i \} + L. \label{thm_sufft_ineq2}
\end{align}
Obviously, the condition $(\E[A^2]+\eta K_{\max})$ $<$ $1$ guarantees that both $\hG_{i}$ and $\bG_{i}$ are bounded from avove when $i$ $\rightarrow$ $\infty$. Since $\hG_{i}$ and $\bG_{i}$ are nondecreasing sequence (refer to the proof of Theorem \ref{thm_PLQinf}), we can conclude that both $\hG_{i}$ and $\bG_{i}$ converge to some finite numbers.
\end{IEEEproof}

\section{Application in Dynamic MV Portfolio Selection }\label{se_mv}
In this section, we illustrate how to apply the derived results in Section \ref{se_sol_P} to solve the dynamic MV portfolio selection problem. We use the similar notations as given in \cite{GaoLiCuiWang:2015}. An investor enters the market with initial wealth $x_0$ and invests simultaneously in $n$ risky assets and one risk free asset. Suppose that the investment horizon is $T$ periods, and the investor decides his/her adaptive investment policy at $t=0, 1, \cdots, T-1$. Denote the deterministic risk free rate as $r_t$, and the random return vector of risk assets in period $t$ as $\e_t$$\triangleq$ $ \left(
e_t^1,  e_t^2,  \cdots,  e_t^n \right)$.\footnote{If $S_t^i$ is the price of $i$-the asset at time $t$, then $e^i_t$ $=$ $S^i_{t+1}/S^i_t$.}  All the randomness is modeled by a standard probability space $\{\Omega,\{ \mathcal{F}_t\}|_{t=0}^T,\mathbb{P} \}$, where $\Omega$ is the event set, $\mathcal{F}_t$ is the $\sigma$-algebra of the events available at time $t$, and $\mathbb{P}$ is the probability measure.\footnote{At time $t$, $t=1,\cdots,T$, the filtration $\cF_t$ is the smallest sigma algebra generated by the realization of $\e_0,\e_1,\cdots, \e_{t-1}$.}  We use the notations $\E_t[\cdot ]$ to denote the conditional expectation $\E[\cdot |\mathcal{F}_t ] $. Let $u_t^i$ be the dollar amount invested in the $i$-th risky asset in period $t$, $i=1,\cdots,n$, and $x_t$ be the wealth level in period $t$. Then, the dynamics of the wealth evolves according to the following stochastic difference equation,
\begin{align}
x_{t+1} =r_tx_t+\P_t\u_t, ~~t=0,\cdots,T-1,\label{def_wealth_process}
\end{align}
where  $\u_t$ $\triangleq$ $\left(u_t^1, u_t^2, \cdots, u_t^n
                     \right)
^{\prime}$ and $\textbf{P}_t$ $\triangleq$ $(P_t^1$,$P_t^2$, $\cdots$, $P_t^n)$ $=$ $(e_t^1-r_t$, $e_t^2-r_t$, $\cdots$, $e_t^n-r_t)$
is the excess return vector. Denote $\gamma_t$$:=$$\prod_{k=t}^{T-1} r_k $ for $t=0,\cdots T-1$ and $\gamma_T=1$ as the discount factor at time $t$. The investor then adopts the following mean-variance formulation to guide his investment,
\begin{align}
(\MV):~ \min_{\{\u_t\}\mid_{t=0}^{T-1}}~&~ \Var_0[x_T]
+ \sum_{t=0}^{T-1}\E_0\left[\u_t^{\prime}\R_t\u_t \right]\label{mv_obj}\\
\textrm{s.t.}~~& \E[{x}_T]=x_{d}, \notag\\
                     ~~&x_{t+1} =r_tx_t+\P_t\u_t, ~~t=0,\cdots,T-1,\notag\\
                     ~~& \u_t\geq0, ~~t=0,\cdots,T-1, \label{mv_noshort}
\end{align}
where $x_d$ is a pre-given target wealth level, $\Var_0[x_T]$ $:=$ $\E_0[x_T^2]$$-$$\E_0[x_T]^2$ is the variance of the terminal wealth and the constraint (\ref{mv_noshort}) implies no short-selling. The second term in objective function (\ref{mv_obj}) is a penalty for the investment in the risky asset, e.g., the investor wants to control his/her wealth exposed to the risky assets.
\begin{assumption}\label{asmp_mv}
 We assume $x_d$ $>$ $\gamma_0 x_0 $.
\end{assumption}
Assumption \ref{asmp_mv} enforces the target wealth level $x_d$ to be greater than the terminal wealth level resulted from investing the whole initial wealth in the risk free asset. Compared with the models studied in \cite{CuiGaoLiLi:2014} and \cite{CuiLiLi:2015}, our model allows correlated returns $\e_t$  over time, no sign constraint on $\E[\P_t]$\footnote{In \cite{CuiGaoLiLi:2014}, the assumption of $\E[\P_t]>0$ is critical to derive the solution.}, and also the inclusion of the penalty term $\u_t^{\prime}\R_t\u_t$ for a purpose of controlling risk exposure.

To solve problem $(\MV)$, we first reformulate it as a special case of the model $(\cP_{\LQ}^T)$. Since the variance is not separable in the sense of dynamic programming, we adopt the embedding method first proposed in \cite{LiNg:2000}. More specifically, introducing the Lagrange multiplier $\lambda\in \mathbb{R}$ for $ \E[{x}_T]=x_d$ yields the following Lagrange function,
\begin{align*}
&\Var_0[x_T]+ 2\lambda(\E[x_T]-x_d)+ \sum_{t=0}^{T-1}\E_0\left[\u_t^{\prime}\R_t\u_t \right]\\
&=\E_0[(x_T-x_d)^2+2\lambda(x_T-x_d)]+\sum_{t=0}^{T-1}\E_0\left[\u_t^{\prime}\R_t\u_t \right],
\end{align*}
which further leads to the following auxiliary problem,
\begin{align*}
(\widehat{\MV}(\lambda)):&\min_{\{\u_t\}\mid_{t=0}^{T-1}}\E_0\Big[(x_T -(x_d-\lambda))^2\Big]
+\sum_{t=0}^{T-1}\E_0\left[\u_t^{\prime}\R_t\u_t \right]\\
&~~~~~~~~~~~~~~~~~~~~~~~~~~~~~~~~~~~~~~~~-\lambda^2\\
\textrm{s.t.}~&~x_{t+1} =r_tx_t+\P_t\u_t, ~~t=0,\cdots,T-1,\notag\\
  ~& ~\u_t\geq0, ~~t=0,\cdots,T-1.
\end{align*}
Replacing the state variable $x_t$ in $(\widehat{\MV}(\lambda))$  by
$x_t$$=$$w_t + (x_d-\lambda)/\gamma_t$ gives rise to the following equivalent problem,
\begin{align*}
(\overline{\MV}(\lambda)): \min_{\{\u_t\}\mid_{t=0}^{T-1}}~&\E_0\big[w_T^2\big]
+\sum_{t=0}^{T-1}\E_0\left[\u_t^{\prime}\R_t\u_t \right]\\
\textrm{s.t.}~&w_{t+1} =r_tw_t+\P_t\u_t, ~~t=0,\cdots,T-1,\notag\\
  ~& \u_t\geq0, ~~t=0,\cdots,T-1.
\end{align*}
It is not hard to see that problem $(\overline{\MV}(\lambda))$ is just a special case of problem $(\cP_{\LQ}^T)$ by letting $A_t=r_t$, $\B_t=\P_t$, $q_t=0$, $\S_t=\0_{n\times 1}$ for $t$ $=$ $0, \cdots, T-1$ and $q_T=1$. Moreover, we introduce the following two mappings: $\hg_t^{\MV}(\K,y,z):$ $\bR^n \times \bR_+ \times \bR_+$ $\rightarrow$ $\bR$ and $\bg_t^{\MV}(\K,y,z):$ $\bR^n \times \bR_+ \times \bR_+$ $\rightarrow$ $\bR$,
\begin{align}
&\hg_t^{\MV}(\K, y, z):=\E_t \Big[ \K^{\prime}\R_t \K+(r_t+\P_t \K)^2\notag\\
&~~~~~~~~~\times \big( y\1_{\{r_t+\P_t \K \geq 0\}}
+ z\1_{\{r_t+\P_t \K <0\}}\big)\Big], \label{def_mv_hgt}\\
&\bg_t^{\MV}(\K,y,z):=
	\E_t \Big[\K^{\prime}\R_t \K +(r_t-\P_t  \K)^2\notag\\
&~~~~~~~~~\times\big(y\1_{\{r_t-\P_t \K \leq 0\}}
+z\1_{\{r_t-\P_t \K > 0\}} \big)\Big].\label{def_mv_bgt}
\end{align}
As the same as for problem $(\cP_{\LQ}^T)$, we further define the following two random variables recursively,
\begin{align}
\hG^{\MV}_t&=\min_{\K \geq 0}~\hg_t^{\MV}( \K, \hG^{\MV}_{t+1}, \bG^{\MV}_{t+1}), \label{def_hatGt_mv}\\
\bG^{\MV}_t&=\min_{\K \geq 0}~\bg_t^{\MV}( \K, \hG^{\MV}_{t+1}, \bG^{\MV}_{t+1}), \label{def_barGt_mv}
\end{align}
with boundary conditions of $\hG^{MV}_{T}=1$ and $\bG^{MV}_{T}=1$. We have the following properties for the sequences $\{\hG^{\MV}_k\}|_{k=0}^T$ and $\{\bG^{\MV}_k\}|_{k=0}^T$.
\begin{lemma}\label{lem_hbG}
For any $t=0,1,\cdots, T-1$,
\begin{align*}
&\hG_t^{\MV}\leq r_t^2\E_t\big[ \hG^{\MV}_{t+1}\1_{\{\hG^{\MV}_{t+1}\geq \bG^{\MV}_{t+1}\}}+\bG^{\MV}_{t+1}\1_{\{\hG^{\MV}_{t+1}< \bG^{\MV}_{t+1}\}}\big],\\
&\bG_t^{\MV}\leq r_t^2\E_t\big[ \hG^{MV}_{t+1}\1_{\{\hG^{\MV}_{t+1}\geq \bG^{\MV}_{t+1}\}}+\bG^{\MV}_{t+1}\1_{\{\hG^{\MV}_{t+1}< \bG^{\MV}_{t+1}\}}\big].
\end{align*}
Furthermore, $\hG^{\MV}_0< \gamma_0^2$ and $\bG^{\MV}_{0}<\gamma_0^2$.
\end{lemma}
We omit the  proof of Lemma \ref{lem_hbG}, as we can prove it in a  similar way to the proof of Theorem \ref{thm_sufft} by using the Lagrange duality theory.  The following theorem characterizes the solution of problem $(\MV)$.

\begin{theorem}\label{thm_MV}
The following policy solves problem $(\MV)$,
\begin{align}
\u_t^*=\begin{dcases}
\Big(x_t-\frac{x_d-\lambda^*}{\gamma_t} \Big)\hK^{\MV}_t &~\textrm{if}~~x_t-\frac{x_d-\lambda^*}{\gamma_t}\geq 0,\\
-\Big(x_t-\frac{x_d-\lambda^*}{\gamma_t} \Big)\bK^{\MV}_t &~\textrm{if}~~x_t-\frac{x_d-\lambda^*}{\gamma_t}<0,
\end{dcases}\label{thm_MV_u}
\end{align}
where $\hK^{\MV}_t$ and $\bK^{\MV}_t$ are computed by
\begin{align*}
&\hK^{\MV}_t=\arg \min_{\K_t \geq 0 }~~\hg^{\MV}_t(\K_t, \hG^{\MV}_{t+1}, \bG^{\MV}_{t+1}),\\
&\bK^{\MV}_t=\arg \min_{\K_t \geq 0 }~~\bg^{\MV}_t(\K_t, \hG^{\MV}_{t+1}, \bG^{\MV}_{t+1}),
\end{align*}
with $\{\hG^{\MV}_t\}|_{t=0}^{T}$ and $\{\bG^{\MV}_t\}|_{t=0}^{T}$ being calculated from (\ref{def_hatGt_mv}) and (\ref{def_barGt_mv}), respectively. The optimal Lagrange multiplier $\lambda^*$ is
\begin{align}
\lambda^*=\frac{\bG^{\MV}_0(x_d-\gamma_0 x_0) }{\bG^{\MV}_0-\gamma_0^2}.\label{thm_MV_lambda}
\end{align}

\end{theorem}

\begin{IEEEproof}
For any fix $\lambda$, applying Theorem \ref{thm_PLQ} generates the following optimal policy for $(\overline{\MV}(\lambda))$,
\begin{align}
\u^*_t(\lambda)= w_t (\hK^{\MV}_t \1_{\{w_t \geq 0\}}- \bK^{\MV}_t \1_{\{w_t < 0\}}).\label{thm_MV_ut(lmd)}
\end{align}
The remaining task is to find the optimal Lagrange multiplier $\lambda$, which can be identified by solving the dual problem, $\lambda^*=\max_{\lambda \in \R^n}~v(\widehat{\MV}(\lambda))$. Applying Theorem \ref{thm_PLQ}, while replacing $w_0$ by $x_0$, yields
\begin{align*}
&\gamma_0^2 v(\widehat{\MV}(\lambda) )=(\gamma_0x_0-x_d +\lambda)^2\big(\hG_0^{\MV}\1_{\{ \lambda \geq x_d-x_0\gamma_0\}}\notag\\
&~~~~~~~~~~~+\bG^{\MV}_0\1_{\{ \lambda < x_d-x_0\gamma_0\}}\big)-\gamma_0^2\lambda^2\notag\\
&=\begin{dcases}
\lambda^2(\hG^{\MV}_0-\gamma_0^2)+2\hG^{\MV}_0\lambda(\gamma_0x_0-x_d)\\
~~~~~~~~+\hG^{\MV}_{0}(\gamma_0x_0-x_d)^2&\textrm{if}~\lambda\geq x_d-x_0\gamma_0,\\
\lambda^2(\bG^{\MV}_0-\gamma_0^2)+2\bG^{\MV}_0\lambda(\gamma_0x_0-x_d)\\
~~~~~~~~+\bG^{\MV}_{0}(\gamma_0x_0-x_d)^2 & \textrm{if}~\lambda< x_d-x_0\gamma_0.\\
\end{dcases}
\end{align*}
From Lemma \ref{lem_hbG} we know that $\hG^{\MV}_0\leq \gamma_0^2$ and $ \bG^{\MV}_0 \leq \gamma^2_0$, which implies $v(\widehat{\MV}(\lambda))$ is a piece-wise concave function of $\lambda$. It is not hard to find the optimal $\lambda^*$ as given in (\ref{thm_MV_lambda}). Then replacing $w_t$ by $x_t$ in (\ref{thm_MV_ut(lmd)}) with $\lambda^*$ yields (\ref{thm_MV_u}).
\end{IEEEproof}
As for the mean-variance efficient frontier, substituting $\lambda^*$ back to  $v(\widehat{\MV}(\lambda))$ gives rise to
\begin{align*}
&\Var[x^*_T]= v(\widehat{\MV}(\lambda^*)-\sum_{k=0}^{T-1}\E_0[(\u^*_t)^{\prime}\R_t\u^*_t] \\
&= \frac{ \bG^{\MV}_0 ( \E[x_T^*] -x_0\gamma_0)^2 }{ \gamma_0^2-\bG^{\MV}_0}-\sum_{k=0}^{T-1}\E_0[(\u^*_t)^{\prime}\R_t\u^*_t].
\end{align*}
Note that, in the above expression, the second term does not have analytical expression. However, it can be evaluated by the Monte Carlo simulation method, once we compute all $\bK^{\MV}_t$ and $\hK^{\MV}_t$.

\section{Examples and Application}\label{se_exampl}
In this section, we first provide a few examples to illustrate our solution procedure for solving $(\cP_{\LQ}^{T})$ and $(\cP_{\LQ}^{\infty})$. Then, we consider a real application of $(\cP_{\LQ}^{T})$
to solve the dynamic portfolio optimization problem.

\subsection{Illustrative Examples of LQ Model}\label{sse_example_control}
\begin{example}\label{exam_control}
We first consider a simple example of $(\cP_{LQ}^{T})$ with $n=3$ and $T=5$. The cost matrices are
\begin{align*}
\R_t=\left(
       \begin{array}{ccc}
         1.2 & 0.3 & 0.4 \\
         0.3 & 1.4 & -0.3 \\
         0.4 & -0.3 & 1.9 \\
       \end{array}
     \right),~~\S_t=\left(
                      \begin{array}{c}
                        -0.2 \\
                        0.6 \\
                        -0.5\\
                      \end{array}
                    \right)
\end{align*}
and $q_t=1.1$ for $t=0,1,\cdots, 4$ with $q_5=1$. We consider two kinds of uncertain system parameters, namely, the independent identically distributed case and a correlated Markovian case.

\textbf{Case 1:} In the first case, we assume $A_t$ and $\B_t$, $t=0,1,2,3,4$, follow the identical discrete distribution with five scenarios as follows,
\begin{align}
&A_t\in \{-0.7,~-0.6,~0.9,~1,~ 1.1\}, \label{exam_control_A}\\
&\B_t \in \Big\{\left(\begin{array}{ccc}
	        	0.18 \\    -0.05 \\    -0.14
            	\end{array}
          \right),
          \left(\begin{array}{ccc}
		    	0.03 \\   -0.12 \\    -0.03
		    	\end{array}
	      \right),
		  \left(\begin{array}{ccc}
	           -0.05 \\     0.05 \\    0.05
	            \end{array}
          \right)\notag\\
	      &\left(\begin{array}{ccc}
	          -0.01 \\     0.05 \\     0.01
	            \end{array}
	      \right),
          \left(\begin{array}{ccc}
	          -0.05 \\     0.01 \\     0.06
                \end{array}
	      \right) \Big\},\label{exam_control_B}
\end{align}
 each of which has the same probability $0.2$. We also consider the following control constraint, $\underline{\d}_t |x_t|$ $\leq$ $\u_t$ $\leq$ $\overline{\d}_t|x_t|$ with
$\underline{\d}_t=\left(\begin{array}{ccc}
 0.1 \\ 0.1 \\ 0.1
 \end{array} \right)$, $\overline{\d}_t=
\left(\begin{array}{ccc}
0.5 \\ 0.5 \\ 0.5
\end{array} \right)$.
By using Theorem \ref{thm_PLQ}, we can identify the optimal control of problem $(\cP_{\LQ}^T)$ as
$\u_t^*(x_t)$ $=$ $\hat{\K}_t x_t \1_{\{x_t \geq 0\}}$ $-$ $\bar{\K}_t x_t \1_{\{x_t < 0\}}$, $t$ $=$ $0$,$\cdots$,$4$, where $\hK_t$ and $\bK_t$ are specified as follows,
\small
\begin{align*}
&\hK_0 =\left(\begin{array}{ccc}
      0.216 \\
      0.100 \\
      0.158 \\
 \end{array}
 \right),
\hK_1=
 \left(\begin{array}{ccc}
 	  0.201 \\
 	  0.100 \\
 	  0.169 \\
 \end{array}
 \right),
\hK_2=
 \left(\begin{array}{ccc}
 	  0.179 \\
      0.100 \\
      0.183 \\
 \end{array}
 \right),\\
&\hK_3=\left(\begin{array}{ccc}
      0.149 \\
      0.100 \\
      0.203 \\
 \end{array}
 \right),
\hK_4=
  \left(\begin{array}{ccc}
      0.108 \\
      0.100 \\
      0.231 \\
 \end{array}
 \right),
\end{align*}
\begin{align*}
&\bK_0 = \left(\begin{array}{ccc}
      0.100 \\
      0.5 \\
      0.100\\
\end{array}
\right),
\bK_1=\left(\begin{array}{ccc}
      0.100 \\
      0.500 \\
      0.100 \\
\end{array}
\right),
\bK_2=\left(\begin{array}{ccc}
      0.100 \\
      0.496 \\
      0.100 \\
\end{array}
\right),\\
&\bK_3=\left(\begin{array}{ccc}
      0.100 \\
      0.480 \\
      0.100 \\
\end{array}
\right),
\bK_4=\left(\begin{array}{ccc}
      0.100 \\
      0.458 \\
      0.100 \\
\end{array}
\right),
\end{align*}
\normalsize
with $\hG_0=3.474$, $\hG_1=3.240$, $\hG_2=2.920$, $\hG_3=2.482$, $\hG_4=1.881$,
$\hG_5=1$ and $\bG_0=3.250$, $\bG_1=3.030$, $\bG_2=2.729$, $\bG_3=2.319$, $\bG_4=1.760$, $\bG_5=1$.
Furthermore, the optimal cost is $V_0(x_0)$ $=$ $3.474 x_0^2 \1_{\{x_0 \geq 0\}}$ $+$ $3.250 x_0^2 \1_{\{x_0 < 0\}}$.

We then extend the control horizon to $T=\infty$ and consider the problem $(\cP_{\LQ}^{\infty})$. From Algorithm \ref{alg:hGbG} and Theorem \ref{thm_PLQinf}, we can compute $\hG^{*}$ $=$ $4.111$ and $\bG^{*}$ $=$ $3.856$, which solve the extended Riccati equations (\ref{def_hatg}) and (\ref{def_barg}). The corespondent optimal stationary control is
\begin{align*}
\u_t^{\infty}(x_t)=\hK^* x_t \1_{\{x_t \geq 0\}} - \bK^* x_t \1_{\{x_t < 0\}},
\end{align*}
for all $t=0,\cdots, \infty$, where
\begin{align*}
\hK^*=\left(\begin{array}{ccc}
    0.259 \\
    0.100 \\
    0.130 \\
\end{array}\right),~\bK^*=\left(\begin{array}{ccc}
    0.100 \\
    0.500 \\
    0.100 \\
\end{array}\right).
\end{align*}
In Figure \ref{fig_examIID}, Sub-figure (a)  plots the outputs of $\hat{G}^*$ and $\bar{G}^*$ with respect to the iteration in Algorithm \ref{alg:hGbG} and Sub-figure (b) plots the state trajectory of 100 sample pathes by implementing the stationary control $\u^{\infty}_t(x_t)$. We can observe that $x_t^*$ converges to $0$ very quickly and the correspondent closed loop system is asymptotically stable.

\begin{figure}
  \centering
  \includegraphics[width=250pt]{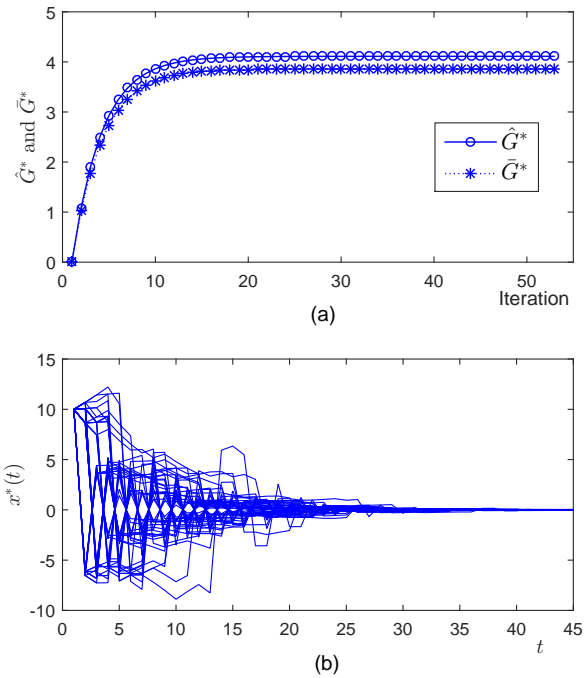}\\
  \caption{Subfigure (a) plots the outputs $\hG^*$ and $\bG^*$ from Algorithm \ref{alg:hGbG}. Subfigure (b) plots the state trajectory for 100 samples of the simulation by implementing the stationary optimal control $u_t^*$ in Example \ref{exam_control}}\label{fig_examIID}
\end{figure}

\textbf{Case 2:} In the previous case, the random matrices $A_t$ and $\B_t$ are independent over time. Now, we consider a simple case when $A_t$ and $\B_t$ are correlated between consecutive periods. Although $A_t$ and $\B_t$ are still assumed to take the values in the $5$ scenarios given in (\ref{exam_control_A}) and (\ref{exam_control_B}),  the scenarios transit among themselves according to a Markov chain with one-step probability,
\begin{align}
P=\left(
    \begin{array}{ccccc}
     0.1 & 0.2 & 0.4 & 0.2 & 0.1\\
     0.4 & 0.1 & 0.3 & 0.1 & 0.1\\
     0.2 & 0.2 & 0.1 & 0.2 & 0.3\\
     0.1 & 0.4 & 0.1 & 0.1 & 0.3\\
     0.1 & 0.2 & 0.1 & 0.5 & 0.1
    \end{array}
  \right).\label{exam_control_prob}
\end{align}
Note that under such a situation, when we compute $\{\hG_t\}|_{t=0}^{T}$ and $\{\bG_t\}|_{t=0}^T$ by (\ref{def_hatGt}) and (\ref{def_barGt}), we actually need to compute the conditional expectation for different scenarios (see the definitions in (\ref{def_hgt}) and (\ref{def_bgt})). Thus, we use the notations $\hG_{t}(j)$ and $\bG_t(j)$ to denote the outputs of (\ref{def_hatGt}) and (\ref{def_barGt}) for scenario $j$ $\in$ $\{1,2,3,4,5\}$. Table \ref{table_exam_LQ} provides the values of $\hG_t(j)$ and $\bG_t(j)$ for all $t$$=$$0,1,2,3$ and $j$ $=$ $1,2,3,4,5$ (We do not list $\hG_4(j)$ and $\bG_4(j)$ since $\hG_4(j)$ $=$ $\bG_4(j)$ $=$ $q_4$ for all $j$).

\begin{table}
  \centering
  \begin{tabular}{ccccc}
     \hline
      $j$   & $\hG_0(j),\bG_0(j)$ & $\hG_1(j),\bG_1(j)$ & $\hG_2(j),\bG_2(j)$ & $\hG_3(j),\bG_3(j)$ \\
      \hline\hline
    j=1 & (3.236, 3.010) & (2.915, 2.711) & (2.473, 2.302) & (1.873, 1.747) \\
    j=2 & (3.029, 2.842) & (2.742, 2.568) & (2.348, 2.193) & (1.805, 1.684) \\
    j=3 & (3.353, 3.155) & (3.017, 2.835) & (2.556, 2.400) & (1.921, 1.802) \\
    j=4 & (3.147, 2.928) & (2.840, 2.643) & (2.422, 2.254) & (1.844, 1.720) \\
    j=5 & (3.339, 3.113) & (3.005, 2.803) & (2.548, 2.379) & (1.927, 1.803) \\
     \hline
     \hline
   \end{tabular}
  \caption{The output $\hG_t(j)$ and $\bG_t(j)$ from iteration (\ref{def_hatGt}) and (\ref{def_barGt}) when the process $A_t$ and $\B_t$ follow a Markov Chain in Example \ref{exam_control} }\label{table_exam_LQ}
\end{table}

\end{example}

\subsection{Application to Dynamic Mean-Variance Portfolio Selection}
\begin{example}\label{exam_MV}
We consider a similar market setting given in \cite{CuiGaoLiLi:2014}, where there are one riskless asset and three risky assets and the market does not allow shorting. The initial wealth is $x_0 = 100$ and the investment horizon is $T$$=$$4$. The return rate of the riskless asset is $r_t = 1$ for $t = 0,1,2,3$. The access return $\P_t$ takes one of the 5 states given in (\ref{exam_control_B}). Different from the setting given in \cite{CuiGaoLiLi:2014}, which assumes independent $\{\P_t\}|_{t=0}^{T-1}$ over time, we assume that $\{\P_t\}|_{t=0}^{T-1}$ evolves according to a Markov Chain with the one-step transition probability matrix given in (\ref{exam_control_prob}). More specifically, if the current state is $i\in\{1,2,\cdots, 5\}$ at  time $t$, in the next time period $t+1$, the excess return $\P_{t+1}$ could be one the $5$ scenario given in (\ref{exam_control_B}) with the probability specified in the $i$th row of the matrix $P$ in (\ref{exam_control_prob}). The penalty matrix $\R_t$ is set as $\R_t=10^{-6} \I_{3}$ for $t$ $=$ $0,\cdots,3$.

Suppose the expected target wealth level set at $x_d=x_0(1+6\%)=106$. By using Theorem \ref{thm_MV}, we can compute $\lambda^*=-0.7824$  and the optimal portfolio policy as
\begin{align*}
\u_t^*(i)=\begin{dcases}
\Big(x_t-106.78 \Big)\hK^{\MV}_t(i) &~\textrm{if}~~x_t-106.78\geq 0,\\
-\Big(x_t-106.78\Big)\bK^{\MV}_t(i) &~\textrm{if}~~x_t-106.78<0,
\end{dcases}
\end{align*}
for $i\in \{1,\cdots, 5\}$, $t=0,1,2,3$. Since $\{\P_t\}|_{t=0}^{T-1}$ are correlated over time, in the optimal portfolio policy, vectors $\hK_t^{\MV}(i)$ and $\bK_t^{\MV}(i)$ depend on the current state $i$ of the Markov Chain. Due to the page limit we only list the solutions $\hK_t^{\MV}(i)$  and $\bK_t^{\MV}(i)$, at time stage $t=0$, for $i=1,\cdots,5$, as follows,
\begin{align*}
&\hK_0^{MV}(i)\in \big\{\left(\begin{array}{c}
                         0.307
                          \\
                         0 \\
                         0 \\
                       \end{array}
                     \right) \left(
                                    \begin{array}{c}
                                      0 \\
                                      0 \\
                                      4.29 \\
                                    \end{array}
                                  \right) \left(
        \begin{array}{c}
          0 \\
          4.41 \\
          0 \\
        \end{array}
      \right)\\
&~~\left(
                        \begin{array}{c}
                          2.03 \\
                          7.15 \\
                          0 \\
                        \end{array}
                      \right) \left(
         \begin{array}{c}
           0 \\
           0 \\
           1.96 \\
         \end{array}
       \right)  \big\},\\
&\bK_0^{MV}(i)\in \big\{\left(\begin{array}{c}
                         33.66
                          \\
                         0 \\
                         46.30 \\
                       \end{array}
                     \right)
\left(
\begin{array}{c}
                                      41.86 \\
                                      1.467 \\
                                      47.08 \\
                                    \end{array}
                                  \right) \left(
        \begin{array}{c}
          40.48 \\
          0 \\
          46.14 \\
        \end{array}
      \right)\\
&~~~\left(
                        \begin{array}{c}
                          36.33 \\
                          0 \\
                          40.83 \\
                        \end{array}
                      \right) \left(
         \begin{array}{c}
           31.50 \\
           2.92 \\
           34.14 \\
         \end{array}
       \right)  \big\}.
\end{align*}
The derived policy $\u^*_t$ is of a feedback nature, i.e., it dependents on the current state information $i$, the time stage $t$ and current wealth level $x_t$. This policy also goes beyond the classical one derived in \cite{CuiGaoLiLi:2014}, in which it not only assumes the independency of the excess return $\P_t$ over time, but also assumes $\E[\P_t]>0$ for all $t$. These assumptions limit the usage of the results derived in \cite{CuiGaoLiLi:2014}, since $\E[\P_t]$ may not always positive in the real market.

\end{example}

\section{Conclusion}\label{se_conclusion}
In this paper, we have developed the analytical optimal control policy for constrained LQ optimal control for scalar-state stochastic systems with multiplicative noise. Most significantly, our formulation can deal with a very general class of linear constraints on state and control variables, which includes the cone constraints, positivity and negativity constraints, and the state-dependent upper and lower bound constraints as its special case. The novel finding on the state-separation property of this kind of models plays a key role in facilitating the derivation of the optimal control policy off line. Different from the classical LQ control model, in which the Riccati equation characterizes the optimal control policy, for our constrained LQ model, we need to invoke two coupled Riccati equations to characterize the optimal control policy.  Besides the problem with a finite control horizon, we extend our results to problems with infinite horizon. We find that the \text{Uncertainty Threshold Principle} \cite{AthanKuGershwin:1977} fails to apply for problems with control constraints. Instead, we provide an algorithm and a sufficient condition to check whether a stationary control policy exists. We have also showed that under the stationary control policy, the closed-loop system is $L^2$-asymptotically stable. Numerical examples demonstrate the effectiveness in implementing our method to solve the constrained control problems and the dynamic constrained mean-variance portfolio optimization problems.


%

\appendices

\section*{Acknowledgment}

\ifCLASSOPTIONcaptionsoff
  \newpage
\fi



%

\bibliographystyle{IEEEtran}
\bibliography{cone_LQnew}

\end{document}